\documentclass{article}

\usepackage{arxiv}

\usepackage[utf8]{inputenc} 
\usepackage[T1]{fontenc}    
\usepackage{hyperref}       
\usepackage{url}            
\usepackage{booktabs}       
\usepackage{amsfonts}       
\usepackage{nicefrac}       
\usepackage{microtype}      
\usepackage{lipsum}

\usepackage{cite}
\usepackage{amsmath,amssymb,amsfonts}
\usepackage{graphicx}
\usepackage{algorithm}
\usepackage[noend]{algpseudocode}
\usepackage{hyperref}
\usepackage{textcomp}
\usepackage{threeparttablex} 
\usepackage{adjustbox}
\usepackage{mathtools}
\usepackage{caption}
\usepackage{float}
\usepackage{stmaryrd}
\usepackage{epstopdf}
\usepackage{multirow}
\usepackage{pifont}
\usepackage{caption}
\usepackage{subcaption}
\usepackage{xcolor}

\newcommand{\norm}[1]{\left\lVert#1\right\rVert}

\newcommand{\R}{{\mathbb{R}}}

\newcommand{\B}{{\mathcal B}}

\newcommand{\Sy}{{\mathcal{S}}}
\newcommand{\N}{{\mathbb{N}}}

\newcommand{\e}{\mathsf{e}}
\newcommand{\ie}{{\it i.e.}}

\newcommand{\X}{{\mathbf{X}}}

\newcommand{\T}{{\mathbf{T}}}
\newcommand{\So}{{\mathbf{S}}}
\newcommand{\Obs}{{\mathcal{U}}}

\newcommand{\U}{{\mathbf{U}}}

\newcommand{\cen}{{\mathsf{\textbf{c}}}}
\newcommand{\rad}{{\mathsf{\textbf{r}}}}
\newcommand{\ex}{\mathsf{e}}
\newcommand{\cmark}{\ding{51}}%
\newcommand{\xmark}{\ding{55}}%

\newtheorem{theorem}{Theorem}[section]

\newtheorem{assumption}{Assumption}

\newtheorem{proposition}[theorem]{Proposition}
\newtheorem{definition}[theorem]{Definition}
\newtheorem{lemma}[theorem]{Lemma}
\newtheorem{remark}[theorem]{Remark}
\newtheorem{problem}[theorem]{Problem}
\newenvironment{proof}{\par\noindent\textbf{Proof.} }{\hfill$\blacksquare$\par}

\title{Input-Constrained Spatiotemporal Tubes for Safe Navigation of Unknown Euler–Lagrange Systems in Dynamic Environments

}

\author{
 Siddhartha Upadhyay \\
  Department of Cyber-Physical Systems\\
  Indian Institute of Science, Bengaluru, India\\
  \texttt{siddharthau@iisc.ac.in} \\
   \And
 Ratnangshu Das \\
  Department of Cyber-Physical Systems\\
  Indian Institute of Science,, Bengaluru, India\\
  \texttt{ratnangshud@iisc.ac.in} \\
   \And
 Pushpak Jagtap \\
Department of Cyber-Physical Systems\\
  Indian Institute of Science,, Bengaluru, India\\
  \texttt{pushpak@iisc.ac.in} \\
}

\begin{document}
\maketitle

\begin{abstract}
Safe navigation in dynamic environments is challenging when system dynamics are unknown and actuator inputs are limited. Existing methods either rely on accurate models, require online optimization, or do not explicitly account for input constraints. This paper presents a real-time control framework for unknown Euler–Lagrange systems that guarantees finite-time reach-avoid-stay (FT-RAS) specifications while respecting actuator limits. We extend the spatiotemporal tube (STT) framework by incorporating input constraints into the controller design and derive offline-verifiable feasibility conditions that relate the available control authority to the tube design and uncertainty bounds. The resulting framework is approximation-free and computationally efficient, making it suitable for real-time implementation. The proposed approach is validated through simulations on a mobile robot, a quadrotor, and a spacecraft, together with hardware experiments on a mobile robot, demonstrating safe navigation while satisfying actuator constraints.
\end{abstract}


\section{Introduction}
Autonomous systems are increasingly being deployed in safety-critical applications \cite{Safety_critical} such as autonomous driving, aerial robotics, industrial automation, and medical robotics. The deployment of autonomous robotic systems in real-world environments presents several challenges \cite{wijayathunga2023challenges},\cite{yin2024formal}. In practice, system dynamics are often partially known or unknown, the operating environment is dynamic and uncertain, and external disturbances can significantly affect system performance. Moreover, all robotic platforms are subject to actuator input constraints arising from physical limitations on force, torque, velocity, and acceleration \cite{simmons1994structured}. These constraints are particularly important in practice, as neglecting them may lead to infeasible control inputs, actuator saturation, degraded performance, and even safety violations. Therefore, developing real-time control frameworks that explicitly account for input constraints while providing formal safety guarantees remain a challenging problem.

The problem of navigation in dynamic environments has been extensively studied in the literature, with planning-based approaches emerging as one of the most prominent solution paradigms. Classical planning approaches, such as A* and RRT \cite{RRTs}, and reactive methods such as Artificial Potential Fields (APF) \cite{APF}, have been widely adopted for autonomous navigation. However, these methods primarily focus on generating collision-free paths without explicitly considering the system dynamics. Consequently, the resulting trajectories may not be dynamically feasible or trackable under the available actuator input constraints. Moreover, APF-based methods are susceptible to local minima, and both approaches generally lack formal guarantees on safety and task completion while providing limited support for explicitly handling actuator input constraints.
Model-based optimization approaches like MPC \cite{MPC,NMPC} have also been proposed to bridge planning and control. For instance, an event-triggered MPC \cite{yang2023empc} integrated with an adaptive APF achieves simultaneous trajectory tracking and obstacle avoidance while explicitly accounting for input and state constraints. However, such approaches rely on accurate system models and online optimization,  although they ensure recursive feasibility and closed-loop stability, they do not provide formal task-level guarantees for reach-avoid objectives. 

Control Barrier Functions (CBFs) \cite{CBF}, being model-based, have been widely used for safety-critical control. However, the basic implementation of CBFs does not guarantee safety in the presence of actuator constraints due to possible optimization infeasibility \cite{ICCBF}. In \cite{ICCBF},  authors present a framework for designing CBFs under input constraints. But apart from being susceptible to modeling errors, extending this framework to dynamic unsafe sets is nontrivial, as the simultaneous satisfaction of safety and input constraints depends on the choice of the class-$\mathcal{K}$ function and iterative steps of defining the safe sets and would require forward simulation of the systems. 
To address moving obstacles, \cite{barrier} employs a Barrier Lyapunov Function but does not explicitly consider actuator input constraints. Similarly, \cite{tayal2026collision} proposes a Collision Cone Control Barrier Function (C3BF) that incorporates relative obstacle motion into the barrier formulation. However, the approach assumes constant obstacle velocities and does not account for actuator limitations, implicitly requiring sufficient control authority to guarantee safety.
Similarly, \cite{Neural_CBF} presents a learning-based approach for synthesizing CBFs for general nonlinear systems under input constraints. However, the requirement of offline training limits its applicability to real-time deployment in dynamically changing environments. In \cite{EL_CBF}, the authors propose a correct-by-construction CBF design for Euler--Lagrange (EL) systems and derive a feasibility condition that can be enforced by design to guarantee safety. Nevertheless, the approach relies on an accurate system model, is susceptible to modeling errors, and incurs significant computational complexity. 

Symbolic control \cite{tabuada2009verification} is another class of correct-by-construction controller synthesis methods capable of handling known general nonlinear dynamics along with input constraints. However, it suffers from the curse of dimensionality, with the controller synthesis time growing exponentially with the system's state dimension. There have been several efforts to address this limitation through compositionality \cite{saoud2019compositional,saoud2021compositional}, through parallelization \cite{khaled2021omegathreads}, using CBFs \cite{sundarsingh2023scalable}, and approximation-free approaches \cite{VCZ}. Nevertheless, these approaches remain largely limited to offline controller synthesis in static environments.

To address unknown system dynamics, the Spatiotemporal Tube (STT) based framework \cite{das2024prescribed,das2025spatiotemporal} has been employed to synthesize approximation-free, closed-form controllers for solving the Temporal Reach-Avoid-Stay (T-RAS). In \cite{STT_Real}, the authors extended the approach to accomplish the T-RAS task in the presence of dynamic unsafe sets. However, these frameworks do not explicitly account for actuator input constraints, and the resulting control inputs may violate these constraints in challenging scenarios, such as avoiding fast-moving obstacles, performing late evasive maneuvers, or navigating around obstacles near the target. Furthermore, the prescribed-time reaching requirement can further increase the control effort, making input constraint violations more likely.

In summary, existing methods either rely on accurate models and online optimization (MPC, CBF-QP), depend on offline synthesis unsuited to dynamic environments (symbolic control, learned CBFs), or neglect actuation limits altogether (existing STT frameworks). To the best of our knowledge, no existing approach provides formal reach-avoid-stay guarantees for unknown systems in dynamic environments while respecting input constraints. The key contributions of this paper are as follows:
\begin{enumerate}
\item We propose a real-time STT framework for unknown Euler--Lagrange (EL) systems that explicitly accounts for actuator input constraints. Unlike \cite{STT_Real}, the proposed tube dynamics are designed to ensure that the resulting tubes are compatible with the available control authority.

\item We develop an approximation-free, closed-form control law based on bounded transformation functions. The controller always satisfies the actuator input limits and guarantees that the system state remains inside the synthesized spatiotemporal tube.

\item We derive feasibility conditions that relate the available actuator limits, system uncertainty bounds, and tube design parameters. These conditions provide a systematic way to design trackable spatiotemporal tubes and certify that the proposed controller can enforce the desired finite-time reach-avoid-stay specification.

\item We prove that, whenever the feasibility conditions are satisfied, the synthesized tube is forward invariant and the closed-loop system satisfies the assigned finite-time reach-avoid-stay specification.

\item We validate the proposed framework through simulations on mobile robots, quadrotor, and spacecraft systems, together with real-world experiments. The results demonstrate safe real-time navigation while respecting actuator input constraints.
\end{enumerate}

\section{Preliminaries and Problem Formulation}
\subsection{Notations}
For $a,b\in\N$ with $a\leq b$, we denote the closed interval in $\N$ as $[a;b] := \{a, a+1, \ldots, b\}$. A ball centered at $\cen \in \mathbb{R}^n$ with radius $\rad \in \mathbb{R}^+$ is defined as $\mathcal{B}(\cen, \rad) := \{ x \in \mathbb{R}^n \mid \|x - \cen\| \leq \rad \}$. We use $x\circ y$ to represent the element-wise multiplication where $x,y\in \R^n$.  An identity matrix of order $n\in \N$ is represented using $I_n$. $x\uparrow(\downarrow)a$ indicates $x$ approaches $a$ from the left (right) side. For a vector $x = [x_1, \dots, x_n]^\top \in \mathbb{R}^n$, the $\ell_p$-norm is denoted by $\|x\|_p = \left( \sum_{i=1}^{n} |x_i|^p \right)^{1/p}$, for $p \geq 1$, and the Euclidean norm is denoted by $\|x\|$ (or equivalently $\|x\|_2$), defined as $\|x\| = \left( \sum_{i=1}^{n} |x_i|^2 \right)^{1/2}$. The interior of a set $A\subset\R^n$ is denoted by $int(A)$. For a vector $x=[x_1,\ldots,x_n]^\top\in\mathbb{R}^n$, $diag(x)$ denotes the diagonal matrix with diagonal entries $x_1,\ldots,x_n$. For $\mathbf{a},\mathbf{b}\in\mathbb{R}^n$, $\mathbf{a}\preceq\mathbf{b}$ (respectively, $\mathbf{a}\succeq\mathbf{b}$) denotes $a_i\le b_i$ (respectively, $a_i\ge b_i$) for all $i=1,\ldots,n$. All other notation in this paper follows standard mathematical conventions.
\subsection{System Definition}\label{subsec:sys}
Consider an Euler-Lagrange (EL) system $\Sy$ described by:
\begin{align}\label{eqn:EL}
    \Sy:M(x)\ddot x+V(x,\dot x)+G(x)=\tau+d(t),
\end{align}
where $x(t)=[x_1(t),....,x_n(t)]\in\X\subset\R^n$ is the system configuration, $\tau(t)\in\R^n$ is the control input, and $d(t)\in \mathbb{D}\subset\R^n$. The term $M(x)\in \R^{n\times n}$, $V(x,\dot x)\in \R^n$, and $G(x)\in \R^n$ denote the inertial, coriolis/centrifugal, and gravity components, respectively. Note that for the sake of simplicity, we drop the arguments of term $M(x),V(x,\dot x)\text{ and }G(x)$ and instead write it as $M,V\text{ and }G$.

For the EL system $\Sy$ described in \eqref{eqn:EL}, we consider $
\bar \tau\in \R^n $ to be the bound on control input i.e. $|\tau(t)|\leq\bar \tau,\forall t \in \R_0^+$. Although the disturbance $d$ and the system parameters $M,V$, and $G$ are unknown, their boundedness can be used for control design. Accordingly, we make the following assumptions \cite[Chapter 2]{ELbook}, \cite[Chapter 7]{ELbook2}, \cite{ELbounds}:
\begin{assumption}\label{ass:el1}
    The external disturbance d satisfies $-\bar d\preceq d(t)\preceq \bar d, \forall t \in \R_0^+$ where $\bar d\in \R^n$ is a known bound and under the scaling of $M^{-1}$, it satisfies $-\underline{m}_i\bar d\preceq M^{-1}d\preceq  \underline{m}_i \bar d$ for some $\underline{m}_i\in \R^+$.
\end{assumption}
\begin{assumption}\label{ass:el2}
    Given the control bound $\bar \tau,$ there exists a positive constant $\underline{m}\in \R,$ such that $\underline{m}\bar \tau \preceq M(t)^{-1}\bar \tau$
\end{assumption}
\begin{assumption}\label{ass:el3}
    The Coriolis and centrifugal terms $V$ and the gravity vector $G$ satisfy $\underline{V}_M\preceq V_M(x,\dot x)\preceq \overline{V}_M$, where $V_M:=-M^{-1}(V+G)$ are bounded as $\underline{V}_M\preceq V_M\preceq\overline{V}_M$ with $\underline{V}_M,\overline{V}_M\in \R^n$, and $V_M^{max}:=\max(-\underline{V}_M, \overline{V}_M)$.
\end{assumption}

\subsection{Problem Formulation}
Let the state of the system $x(t)$ be subjected to a \textit{finite-time reach-avoid-stay (FT-RAS)} specification defined next:
\begin{definition}[{Finite-Time }Reach-Avoid-Stay (FT-RAS) task]\label{def: tras}
Given the system configuration $x(t)$, a time-varying unsafe set $\U:\R_0^+\rightarrow\R^n$, an initial set $\mathbf{S}\subset\X\setminus\U(0)$ with the given initial condition $x(0)\in \mathbf{S}$, and a target set $\T$ , we say that the state $x(t)$ satisfies the FT-RAS task if:
\begin{align}
      \exists t_c\text{ s.t. }x(t)\in \T,\forall t\in[t_c,\infty),\text{ and } \quad x(t)\notin\U(t), \forall t\in \R_0^+\notag
\end{align}
\end{definition}
\begin{definition}\label{def: obs}

    The unsafe set $\U(t)$ is defined as a union of $n_o$ time-varying balls, $\mathcal{U}^{(j)}(t),\forall j \in [1;n_o]$ defined as follows:
    \begin{align}
        \U(t)=\bigcup_{j=1}^{n_o}\Obs^{(j)}(t)\text{ with }\mathcal{U}^{(j)}(t):=\B(o^{(j)}(t),\rad_o^{(j)}),\notag
    \end{align}
    where $o^{(j)}(t)\in \R^n$ and $\rad^{(j)}_o\in\R^+$ denote the time-varying center and radius capturing the region surrounding the $j^{th}$ dynamic obstacle. We upper bound the rate at which the center of the time-varying obstacle changes, by assuming that its velocity is bounded, i.e.,
\begin{align}
    \norm{\dot{o}^{(j)}(t)} \leq v_o, \quad v_o \in \mathbb{R}^+_0.\notag
\end{align}

\end{definition}
We now state the main control problem addressed in this work:
\begin{problem}[Real-Time Input-Constrained FT-RAS]\label{prob1}
Given the system \eqref{eqn:EL} under Assumptions \ref{ass:el1}--\ref{ass:el3}, subject to prescribed input constraints, and an FT-RAS task as defined in Definition \ref{def: tras}, synthesize a \textit{real-time}, \textit{approximation-free}, and \textit{closed-form} control law that guarantees satisfaction of the FT-RAS specification while respecting the \textit{input constraints} at all times.
\end{problem}

To solve Problem \ref{prob1}, we utilize the STT framework, which defines a safe time-varying region in state space that remains goal-directed.
\begin{definition}\label{def:stt}
    Given a FT-RAS task in Definition \ref{def: tras}, a spatiotemporal tube (STT) $\Gamma(t) = \mathcal{B} (\cen(t), \rad(t))$ is characterized by a time-varying center $\cen: \R_0^+ \rightarrow \mathbb{R}^n$ and radius $\rad: \R_0^+ \rightarrow \mathbb{R}^+$, if the following holds
    \begin{subequations}\label{eqn:stt}
    \begin{align}
        &\rad(t) > 0, \forall t \in \R_0^+,\\
        &\exists t_c\in \R^+ \text{ such that } \Gamma(0) \subseteq \So,\ \Gamma(t) \subseteq \T,\forall t\in[t_c,\infty),\\
        &\Gamma(t) \cap \U(t) = \emptyset, \forall t \in \R_0^+.
    \end{align}
    \end{subequations}
\end{definition}

\begin{remark}
If one designs a control law that constrains the state of the system $x(t)$ within the STT, i.e.,
\begin{align} \label{eqn:stt_con}
    x(t) \in \Gamma(t), \forall t \in \R_0^+,
\end{align}
while respecting the input constraint, then one can ensure satisfaction of the FT-RAS specification for the given EL system in \eqref{eqn:EL}.
\end{remark}

\section{Designing Spatiotemporal Tubes}
In this section, we propose the design of a spatiotemporal tube which would satisfy the FT-RAS task.
\subsection{Preliminaries of STT Design}
We begin by selecting a point $s=[s_1,\ldots,s_n] \in int(\mathbf{S})$ and $\eta=[\eta_1,\ldots,\eta_n] \in int(\T)$ inside the initial set and the target set. Next, we define balls of radius $d_S,d_T\in \R^+$ around the points $s\text{ and } \eta$ and denote it by $\mathbf{\hat S}$ and $\hat\T$, as follows:
\begin{align}
    \mathbf{\hat S}&=\B(s,d_S):=\{x\in \R^n|\norm{x-s}\leq d_S\}\notag\\
    \mathbf{\hat \T}&=\B(\eta,d_T):=\{x\in \R^n|\norm{x-\eta}\leq d_T\}\notag
\end{align}
such that $x(0) \in \mathbf{\hat S}\subset \mathbf{S}$ and $\hat\T\subset \T$. As introduced in Definition~\ref{def:stt}, the STT $\Gamma(t) = \mathcal{B}(\cen(t), \rad(t))$ is defined by a time-varying center $\cen: \R_0^+ \rightarrow \R^n$ and radius $\rad: \R_0^+ \rightarrow \R^+$.
Additionally, to ensure separation between the unsafe set, we make the following assumptions:
\begin{assumption}\label{ass:obs}
   We assume that the unsafe set $\mathcal{U}^{(j)}(t),\forall j \in [1;n_o]$ remains separated from each other by at least a distance $2\rad_{a}$ at all times, where $\rad_a\in \R^+$ and can be considered as an acting radius. 
\end{assumption}
\begin{remark}
The requirement in Assumption~\ref{ass:obs} is imposed primarily for analysis and does not limit the practical applicability of the framework. Whenever two unsafe sets become sufficiently close such that the assumption no longer holds, they may be merged into an equivalent unsafe set that minimally encloses both regions. This approximation would preserve the safety guarantees while ensuring that the assumption remains satisfied.
\end{remark}
\subsection{STT Center Dynamics}
The evolution of the center $\cen(t)$ with, $\cen(0)=s \in int(\So)$ is, governed by the following dynamics:
\begin{align}\label{eqn:cen_dyn}
\dot{\cen}(t) =
\begin{cases}
k_1 \gamma(\cen(t),\eta), &\text{if } d^{(j)}(\cen(t),o^{(j)}(t)) > \rad_{a},\ \forall j \in [1;n_o], \\
k_1 \frac{d^{(j)}(\cen(t),o^{(j)}(t))}{\rad_{a}} \gamma(\cen(t),\eta) 
+ \notag \\
\qquad \left(1 - \frac{d^{(j)}(\cen(t),o^{(j)}(t))}{\rad_{a}}\right)
\left(k_2 m^{(j)} + k_3 v^{(j)}\right), 
& \text{if } {\exists j \in [1;n_o]\text{ }} d^{(j)}(\cen(t),o^{(j)}(t)) \le \rad_{a}.
\end{cases}
\end{align}
The dynamics of the center in \eqref{eqn:cen_dyn} is divided into two cases: 
\begin{enumerate}
    \item[(i)] when no unsafe sets are within the acting radius $\rad_{a} \in \mathbb{R}^+$, and
    \item[(ii)] when an unsafe set is present within the acting region.
\end{enumerate}
In the first case, the unsafe set lies outside the acting radius, i.e.,
\begin{align}
d^{(j)}(\cen(t), o^{(j)}(t)) = \norm{\cen(t) - o^{(j)}(t)} - \rad^{(j)}_o - \rad_{min} > \rad_{a},
\end{align}
where $\rad_{min}$ denotes the minimum allowable distance between the tube center 
and the unsafe set. In this scenario, the vector  
$\gamma(\cen(t), \eta)$ drives the tube center toward the target point $\eta$ and is defined as follows:
\begin{align}\label{eqn:gamma}
    \gamma(\cen(t),\eta)&=-\frac{(\cen(t)-\eta)^{\frac{1}{3}}}{\sqrt{\norm{(\cen(t)-\eta)^{\frac{1}{3}}}^2+\epsilon^2}},\quad \epsilon\in \R^+
\end{align}
and $k_1\in \R^+$ is the constant that controls the rate at which the center gets pulled toward the target.

The second case is activated whenever an unsafe set enters the acting radius, $\ie$, $d^{(j)}(\cen(t),o^{(j)}(t))=\norm{\cen(t)-o^{(j)}(t)}-\rad^{(j)}_o-\rad_{min}\leq \rad_{a}$. By Assumption~\ref{ass:obs}, at most one unsafe set can lie within the acting radius at any time. In this case, the obstacle avoidance terms $m^{(j)}$ and $v^{(j)}$  associated with the $j$-th unsafe set are activated to ensure that the tube center remains at least $\rad_{\min}\in\mathbb{R}^{+}$ away from the unsafe set $\mathcal{U}^{(j)}$, while the goal-directed term $\gamma(\cen(t),\eta)$ simultaneously drives the tube toward the target. The obstacle avoidance terms $m^{(j)}$ and $v^{(j)}_1$ are defined as:
\begin{align}
    m^{(j)}=\frac{\cen(t)-o^{(j)}(t)}{\norm{\cen(t)-o^{(j)}(t)}}\notag
\end{align}
and $v^{(j)}$ is any unit vector that meets the condition ${m^{(j)}}^\top v^{(j)}=0$, while $k_2\geq v_o,k_3\in \R^+$ are the positive constants that affect the avoidance term.

In addition, to ensure a safe approach to the target, we make the following separation assumptions:
\begin{assumption}\label{ass:reach}
    To guarantee safe convergence to the target, we assume that there exists a time $t_1 \in \mathbb{R}^+$ such that the unsafe set remains at least a distance $\rad_a$ from the STT center $\cen(t)$ for all $t \geq t_1$.
\end{assumption}
\begin{remark}
    Assumption~\ref{ass:reach} is required to guarantee safe target reaching. We only assume that there exists a finite time $t_1$ after which a safe passage toward the target becomes available.

    Importantly, the proposed framework guarantees safety throughout the entire interval $[0,t_1]$ by continuously avoiding the unsafe sets. Once the separation condition in Assumption~\ref{ass:reach} is satisfied after any time $t_1$, the target-driven term dominates the tube dynamics, thereby ensuring finite-time convergence to the target set. If Assumption~\ref{ass:reach} does not hold for an extended period, the proposed framework continues to guarantee safety by avoiding the unsafe sets. Therefore, Assumption~\ref{ass:reach} is only required to ensure that eventual target reachability is feasible.

\end{remark}
\subsection{STT Radius Dynamics}
Next, we define the dynamics of the tube's radius $\rad(t)$, which adapts according to the tube's proximity to unsafe sets, as follows:
\begin{equation}\label{eqn:radius_dynamics}
\dot{\rad}(t)=\frac{\ex^{-\nu d_1^{(j)}{(t)}}\dot{ d}^{(j)}_1(t)}{\big ( \ex^{-\nu\rad_{max}}+\ex^{-\nu d^{(j)}_1(t)}\big)},\quad {\rad(0)=\rad_{max}},
\end{equation}
where $\rad_{max}\leq\min(d_S,d_T,\rad_{a})$ is the maximum allowable tube radius, $\nu \in \R^+$ is an arbitrarily large smoothing parameter, and $d^{(j)}_1(t)$ denotes the distance to the unsafe sets:
\begin{align}\label{eqn:rho}
     d^{(j)}_1(t) &=\| \cen(t) - o^{(j)}(t) \| - \rad_o^{(j)}, \forall j\in [1;n_o],\notag
\end{align}
where $d^{(j)}_1(t)$ is the distance between the tube center $\cen(t)$ and the $j^{th}$ the unsafe set.  Thus, according to Equation \eqref{eqn:radius_dynamics}, the radius of the tube changes from $\rad_{max}$ when the tube is close to any unsafe set, as shown in Figure~\ref{fig:STT_Explain}, the radius shrinks to avoid the collision with the unsafe set and expands when it is farther away from the unsafe sets. 

Thus, given a time-varying center $\cen: \R_0^+ \rightarrow \R^n$ and radius $\rad: \R_0^+ \rightarrow \R^n$, governed by the dynamics in Equations \eqref{eqn:cen_dyn} and \eqref{eqn:radius_dynamics}, we define the STT $\Gamma(t) = \mathcal{B} (\cen(t), \rad(t))$ as a closed ball in $\R^n$ centered at $\cen(t)$ with radius $\rad(t)$:
\begin{equation}\label{eqn:stt_ball}
    \Gamma(t)\! :=\! \{ x\!\in \!\R^n \!\mid \! \|x\!-\!\cen(t)\| \!\leq\! \rad(t)\}, \forall t \!\in \!\R_0^+ .
\end{equation}

\begin{figure}
    \centering
    \includegraphics[width=0.5\linewidth]{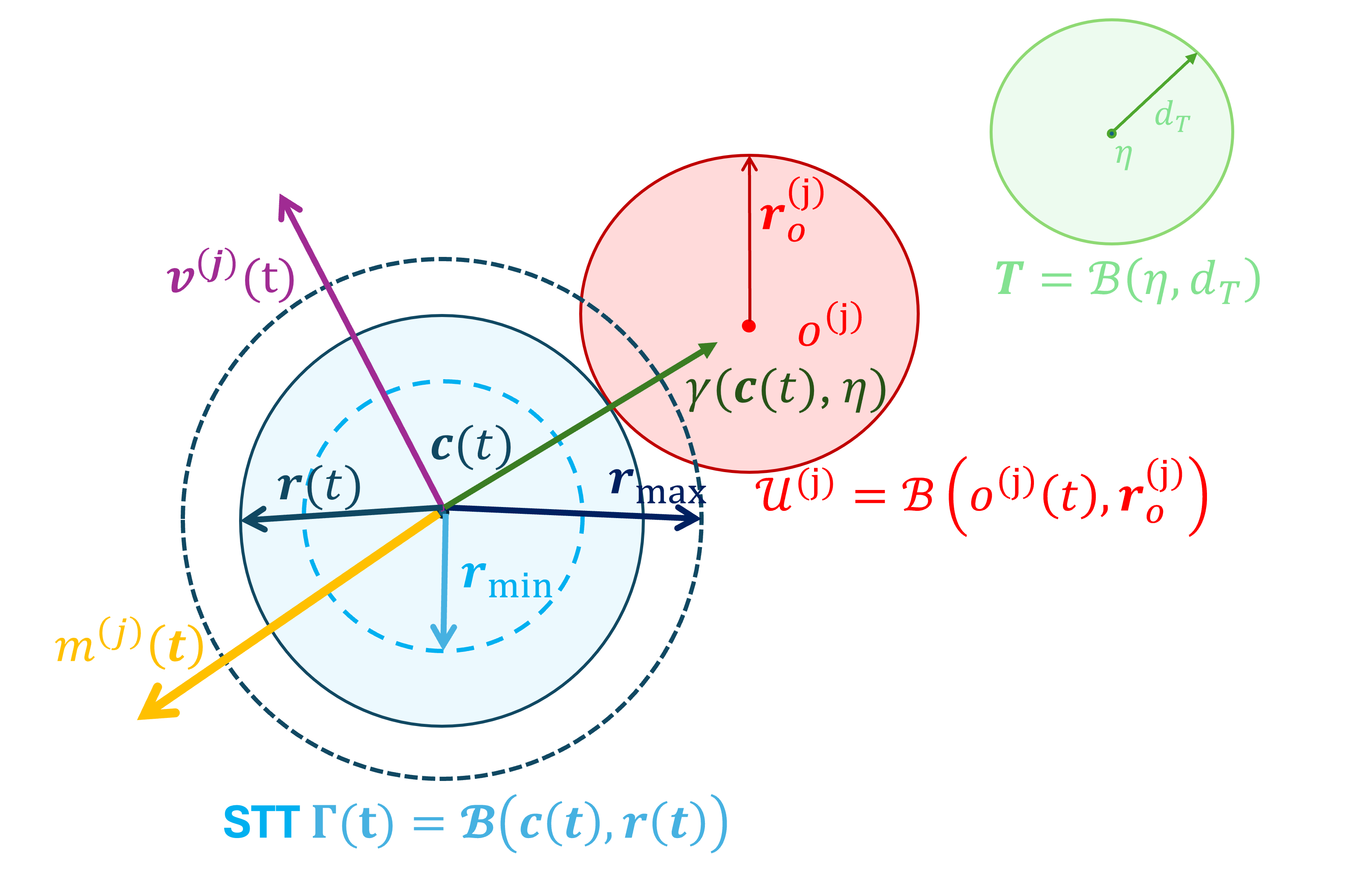}
    \caption{Pictorial representation of different terms affecting the STT center dynamics and radius dynamics.}
    \label{fig:STT_Explain}
\end{figure}

\subsection{Theoretical Guarantee of FT-RAS Satisfaction}
The next theorem guarantees that the designed STT adheres to the following three key conditions for satisfying the FT-RAS specifications. First, the tube reaches the target set within some finite time $t_c\in \R^+$ and remains in it for all further time. Second, the tube remains entirely outside the unsafe set for all times $t \in \R_0^+$, ensuring safety. Finally, the radius of the tube remains strictly positive throughout the motion.
\begin{theorem}\label{theorem_tube}
The STT $\Gamma(t)$ in \eqref{eqn:stt_ball} meets the following requirement to ensure satisfaction of the FT-RAS specification:
\begin{enumerate}
    \item[(i)] Starts from the initial set $\ie$ $\Gamma(0) \subseteq \So$ and reaches the target set within finite time $t_c\in \R^+$ and stays there for all further times: $\Gamma(t) \subseteq \T,\forall t \in [t_c,\infty)$.
    \item[(ii)] Avoids the unsafe set: 
    $\Gamma(t) \cap \U(t) = \emptyset$,  $\forall t \in [0,\infty)$.
    \item[(iii)] The radius stays positive: $\rad(t) \in \mathbb{R}^+, $ $\forall t \in [0,\infty)$.
\end{enumerate}
\begin{proof}
    We will prove each of the claims individually:

    $(i)$ At t=0, the tube center starts from the point $s$, $\cen(0)=s$.
    Also, we can rewrite the solution for the radius dynamics $\rad(t)$ as follows:
    \begin{align}
  \rad(t)=-\frac{1}{\nu}\ln(\e^{-\nu\rad_{max}}+e^{-\nu d^{(j)}_1(t)}), \label{eqn:rad_sol}
\end{align}
which is a smooth approximation of the min operator, the radius of the tube satisfies the inequality:
\begin{align}\rad(t)\leq&\min(\rad_{max},d^{(j)}_1(t))\notag\\
    \implies \rad(t)\leq& \rad_{max}\leq \min(d_S,d_T,\rad_{a}),\forall t\in \R^+_0.\label{eqn:rad_max}
\end{align}
Thus, from \eqref{eqn:rad_max} the STT at $t=0$ can be expressed as $\Gamma(0)=\B(s,\rad(0))\subset \So$.

Next, in order to show that the STT reaches the target set $\T$ and stays there, we use Assumption \ref{ass:reach}, which confirms the existence of $t_1\in \R^+$ such that for all $t>t_1$, the unsafe sets remain a minimum $\rad_{a}$ distance away from the center of the tube $\cen(t)$, implying $d^{(j)}(\cen(t),o_p^{(j)})>\rad_{a},\forall j \in [1;n_o]$. Thus, for all $t>t_1$, Equation \eqref{eqn:cen_dyn} simplifies to:
\begin{align}
    \dot c=-k_1\frac{(\cen(t)-\eta)^{\frac{1}{3}}}{\sqrt{\norm{(\cen(t)-\eta)^{\frac{1}{3}}}^2+\epsilon^2}}.\notag
\end{align}
We substitute $z=(\cen(t)-\eta)$ and rewrite the above equation as $ \dot z=-k_1\frac{z^{\frac{1}{3}}}{\sqrt{\norm{z^{\frac{1}{3}}}^2+\epsilon^2}}$. Next, we consider a positive definite and radially unbounded Lyapunov function as $V=\frac{1}{2}z^\top z=\frac{1}{2}\|z\|^2,$ and then we compute $\dot V$:
\begin{align}
    \dot V=&\frac{-k_1 z^\top z^{\frac{1}{3}}}{\sqrt{\|z^{\frac{1}{3}}\|^2+\epsilon^2}}.\notag
\end{align}
Next, using norm equalities and inequalities 
$z^\top z^{\frac{1}{3}}=\sum_{i=1}^{i=n}z_i^{\frac{4}{3}}=\|z\|^{\frac{4}{3}}_{\frac{4}{3}}\geq c_1\|z\|^\frac{4}{3}$ and $
    \|z^\frac{1}{3}\|^2= \sum_{i=1}^n z^{\frac{2}{3}}=\|z\|^\frac{2}{3}_{\frac{2}{3}}\leq c_2\|z\|^\frac{2}{3}$, we get
\begin{align}
     \dot V=&-\frac{k_1 \|z\|^{\frac{4}{3}}_{\frac{4}{3}}}{\sqrt{\|z^{\frac{1}{3}}\|^2+\epsilon^2}}\leq-\frac{k_1c_1\|z\|^\frac{4}{3}}{\sqrt{c_2\|z\|^\frac{2}{3}+\epsilon^2}}.\label{eqn:lyapunov_eq}
\end{align}
Next, we substitute $2V=\|z\|^2$ in \eqref{eqn:lyapunov_eq} and we get:
\begin{align}
    \dot V\leq-\frac{k_1'V^\frac{2}{3}}{\sqrt{c_1'V^\frac{1}{3}+\epsilon^2}},\notag
\end{align}
where $k_1'=k_1c_12^\frac{2}{3}$ and $c_1'=c_22^\frac{1}{3}$, and since $V$ is strictly decreasing with respect to time and a positive definite function, we can ensure the asymptotic convergence of $z\rightarrow0$ and $V$ decreases with time $\ie$, $V(t)<V(t_1),\forall t \in[t_1,\infty)$ \cite[Theorem 4.9]{khalil2002nonlinear}. We use the strictly decreasing property of $V$ to derive the following inequality:
\begin{align}
    V(t)< V(t_1),\forall t \in[t_1,\infty) \implies
    -\frac{1}{\sqrt{c_1'V(t)^\frac{1}{3}+\epsilon^2}}\leq -\frac{1}{\sqrt{c_1'V(t_1)^\frac{1}{3}+\epsilon^2}}.\label{eqn:Veq2}
\end{align}
Substituting the inequality in \eqref{eqn:Veq2} into \eqref{eqn:lyapunov_eq}, we get:
\begin{align}
    \dot V\leq-cV^\frac{2}{3},\label{eqn:final V}
\end{align}
where $c=\frac{1}{\sqrt{c_1'V(t_1)^\frac{1}{3}+\epsilon^2}}$. Comparing the equation in \eqref{eqn:final V} with $\dot V\leq-cV^\alpha,c\in \R^+ \text{ and }\alpha \in (0,1)$, thus by Theorem 4.2 in \cite{bhat2000finite}, we get the time $t_c$ at which $z(t_c)=0$ $\ie$ $c(t_c)=\eta$ and the time $t_c$ is given by $t_c=\frac{3}{c}V(t_1)^\frac{1}{3}+t_1$. Moreover, since all the unsafe sets are at least $\rad_{a}$ distance away from the center of the tube by Assumption \ref{ass:reach}, we can ensure that $c(t)=\eta, \forall t \in [t_c,\infty)$. For radius of tube, we use the inequality from \eqref{eqn:rad_max} and hence we can conclude that:
\begin{align*}
    \Gamma(t)\subset\T,\forall t \in [t_c,\infty).
\end{align*}

$(ii)$ Next, we will prove the second claim. For this, we consider that the $j^{th}$ unsafe set is inside the acting radius $\ie$  $d^{(j)}(t)\leq\rad_{a}$ and let a time-varying function $J^{(j)}$ for $j^{th}$ unsafe set be defined as follows:
\begin{equation}
    J^{(j)}(t)=(\cen(t)-o^{(j)}(t))^\top(\cen(t)-o^{(j)}(t))-(\rad^{(j)}_o+\rad_{min})^2,\nonumber
\end{equation}
which measures the squared distance between the STT center and the center of the $j^{th}$ unsafe set, offset by the safety margin $\rad^{(j)}_o+\rad_{min}$.
The time derivative of $J^{(j)}(t)$ is:
\begin{align}
    \dot{J}^{(j)}(t)&=2(\cen(t)-o^{(j)}(t))^\top\dot{\cen}(t)-2(\cen(t)-o^{(j)}(t))^\top\dot{o}^{(j)}(t). \notag
\end{align}
We analyze $\dot{J}^{(j)}(t)$ on the boundary of the safe margin around each unsafe set, $\|\cen(t) - o^{(j)}(t)\| = \rad_o^{(j)} + \rad_{min}.$
Substituting the expression for $\dot{\cen}(t)$ into $\dot{J}^{(j)}(t)$ yields:
\begin{align}
    \dot J^{(j)}(t)=&2(\cen(t)-o^{(j)}(t))^\top k_1 \frac{d^{(j)}(\cen(t),o^{(j)}(t))}{\rad_{a}} \gamma(\cen(t),\eta)\notag\\
    &+2(\cen(t)-o^{(j)}(t))^\top  \left(1 - \frac{d^{(j)}(\cen(t),o^{(j)}(t))}{\rad_{a}}\right)
\left(k_2 m^{(j)} + k_3 v^{(j)}\right)-2(\cen(t)-o^{(j)}(t))^\top\dot{o}^{(j)}(t). \notag
\end{align}
As $\| \cen(t) - o^{(j)}(t) \| \rightarrow \rad_o^{(j)} + \rad_{min}$, we have $d^{(j)}=0$. Thus, we obtain the following equation for $\dot J^{(j)}(t)$ as follows:
\begin{align}
    \dot{J}^{(j)}(t)=2k_2(\cen(t)-o^{(j)}(t))^\top \frac{\cen(t)-o^{(j)}(t)}{\norm{\cen(t)-o^{(j)}(t)}}+2k_3(\cen(t)-o^{(j)}(t))^\top v^{(j) }
-2(\cen(t)-o^{(j)}(t))^\top\dot{o}^{(j)}(t) .\notag
\end{align}
We use the fact that $(\cen(t)-o^{(j)}(t))^\top v^{(j) }=0$ by construction of $v^{(j)}$, thereby simplifying the expression for $\dot J^{(j)}(t)$ and establishing the following lower bound:
\begin{align}
    \dot J^{(j)}(t)\geq 2\norm{\cen(t)-o^{(j)}(t)}(k_2-\|\dot o^{(j)}\|).\notag
\end{align}

We use the condition $k_2 \ge v_o$ together with the bound $\|\dot{o}^{(j)}\| \leq v_o$ from Definition~\ref{ass:obs} to get that $\dot{J}^{(j)}(t) > 0$ as $\| \cen(t) - o^{(j)}(t) \| \rightarrow \rad_o^{(j)} + \rad_{min}$, the function $J^{(j)}(t)$ cannot decrease to zero or become negative. Therefore, $J^{(j)}(t) > 0$ holds for all $t \in\R_0^+$, implying 
$$\| \cen(t) - o^{(j)}(t) \| > \rad_o^{(j)} + \rad_{min} \ \text{ for all } \ t \in \R_0^+.$$
Hence, the STT center consistently preserves a minimum separation of $\rad_{min}$ from the $j$-th unsafe set at all times. Extending the same reasoning to all $j \in [1; n_o]$, it follows that the STT center remains at a safe distance from every unsafe set for all time.

Next, we need to guarantee that the tube $\Gamma(t)=\B(\cen(t),\rad(t))$ does not intersect with any unsafe set, it suffices to show that:
\begin{align}\label{eqn:radprove}
    \rad(t)\leq d^{(j)}_1(t)=\norm{\cen(t)-o^{(j)}}-\rad_o^{(j)}.
\end{align}
We use the solution of radius dynamics  from \eqref{eqn:rad_sol} and consider the following two cases:

\textbf{Case 1:} $\rad_{max}\leq d^{(j)}_1(t)$:
\begin{align}
    \rad(t)=&-\frac{1}{\nu}\ln(\e^{-\nu\rad_{max}}+e^{-\nu d^{(j)}_1(t)})\notag\\
    \leq&\min(\rad_{max},d^{(j)}_1(t))
    \leq \rad_{max}.\notag
\end{align}
Therefore, for any $j\in[1;n_o]$, we have
$\rad(t) \leq \rad_{max} \leq d^{(j)}_1(t),$ satisfying condition~\eqref{eqn:radprove}.

\textbf{Case 2:} $\exists \hat j \in [1;n_o], \rad_{max} > d^{(\hat j)}_1(t)$ \ie \  $\hat{j}^{th}$ obstacle is inside the acting radius. Again, using the solution for radius dynamics and using the following inequality, we get:
\begin{align}
     \rad(t) &= -\frac{1}{\nu} \ln \left( \ex^{- \nu \rad_{max}} + \ex^{- \nu  d^{(\hat j)}_1(t)} \right)\notag \\
     &\leq\min \Big( d^{(\hat j)}_1(t),\rad_{max}\Big) \leq d^{(\hat j)}_1(t),\notag
\end{align}
thus ensuring condition in \eqref{eqn:radprove} satisfies.
Therefore, in both scenarios \eqref{eqn:radprove} is satisfied, guaranteeing that the STT $\Gamma(t)$ doesn't intersect any unsafe set at any time: $\Gamma(t)\cap\U(t)=\emptyset, \forall t \in \R_0^+$.

$(iii)$ In the second part of the proof we had already established that for any $j\in[1;n_o]$, the STT center $\cen(t)$ maintains a minimum distance of $\rad_{min}$ from each of the unsafe sets for all times:
\begin{align}
    \|\cen(t)-o^{(j)}(t)\|\geq\rad_o^{(j)}+\rad_{min}\implies d^{(j)}_1(t)\geq \rad_{min}.\notag
\end{align}
Substituting this into radius dynamics from \eqref{eqn:rad_sol}, we get:
\begin{align}
    \rad(t) \geq -\frac{1}{\nu} \ln \left( \ex^{- \nu \rad_{max}} + \ex^{- \nu \rad_{min}} \right) > 0, \quad \forall t \in\R_0^+.\notag
\end{align}
\begin{remark}
   One can select $\nu>\frac{\ln 2}{\rad_{min}}$ which will ensure that $\rad(t)>0$, since $\rad_{max},\rad_{min}>0$.
\end{remark}
\end{proof}
\end{theorem}
\subsection{Bounds on STT Center and Radius Dynamics}
This section presents bounds on the center dynamics \eqref{eqn:cen_dyn} and radius dynamics \eqref{eqn:radius_dynamics}, which are later utilized in deriving feasibility conditions.
\begin{lemma}\label{lem:STT_bound}
    The Euclidean norm of time derivative of STT center $\cen(t)$ and radius $\rad(t)$ are upper bound for all time $t\in \R_0^+$ as
    \begin{align}
        \norm{\dot \cen(t)}\leq \bar \cen,\quad
        |\dot \rad(t)|\leq \bar \rad,\label{eqn:rad_bound}
    \end{align}
    with $\bar \cen=\max\Big(k_1,\sqrt{k_2^2+k_3^2}\Big)$ and $\bar \rad=\max\Big(k_1,\sqrt{k_2^2+k_3^2}\Big)+v_o$.
\end{lemma}
\begin{proof}
  \textbf{Part 1:}  In the first part, we establish an upper bound on the center dynamics. 
Considering the center dynamics defined in \eqref{eqn:cen_dyn}, we analyze 
the case when $d(\cen, o_p^{(j)}) > \rad_{a}$. By applying norm inequalities 
and using the fact that $\norm{\gamma(\cen(t), \eta)} \leq 1$, which follows 
from the construction of $\gamma$ in \eqref{eqn:gamma}, we obtain the following 
upper bound on the norm of the time derivative of the tube center:
\begin{align}
    \norm{\dot{\cen}(t)} \leq k_1.\label{eqn:bound1}
\end{align}
  Similarly, for the second case when $d(\cen, o_p^{(j)}) \leq \rad_{a}$, 
we begin by applying norm inequalities. Using the fact that
$\norm{\gamma(\cen(t),\eta)} \leq 1$, $\norm{m^{(j)}} \leq 1$, and 
$\norm{v^{(j)}} \leq 1$ for all $j \in [1;n_o]$, and defining 
$\lambda = \frac{d(\cen,o_p^{(j)})}{\rad_{a}} \in [0,1]$, 
the bound on center dynamics can be rewritten as:
  \begin{align}
      \norm{\dot \cen(t)}&\leq\lambda k_1+(1-\lambda)(k_2m^{(j)}+k_3m^{(j)}\notag)\\
      \implies \norm{\dot \cen(t)}&\leq \max\Big(k_1,\sqrt{k_2^2+k_3^2}\Big).\label{eqn:bound2}
  \end{align}
\end{proof}
Thus, from Eqs. \eqref{eqn:bound1} and \eqref{eqn:bound2}, we can conclude that the time derivative of the center of the tube follows the following upper bound for all times $t\in \R_0^+$:
\begin{align*}
     \norm{\dot \cen(t)}\leq \max\Big(k_1,\sqrt{k_2^2+k_3^2}\Big).
\end{align*}
 \textbf{Part 2:} In this part, we will establish an upper bound on the radius dynamics. For this at first, we find $\dot d^{(j)}(t)$ as follows:
 \begin{align*}
      \dot d^{(j)}(t)=\frac{(\cen(t)-o^{(j)})^\top(\dot\cen(t)-\dot o^{(j)}(t))}{\norm{\cen(t)-o^{(j)}}},\notag
 \end{align*}
 then we use the following upper bound:
 \begin{align}
     \frac{\ex^{-\nu d^{(j)}{(t)}}}{\big ( \ex^{-\nu\rad_{a}}+\ex^{-\nu d^{(j)}(t)}\big)}\leq1,\label{eqn:rad_bound1}\\
     |\dot d^{(j)}(t)|\leq\ \norm{\dot\cen(t)}+\norm{\dot o^{(j)}(t)}.\notag
 \end{align}
  Thus, using the upper bound on $\norm{\dot c(t)}$ from the first part and the upper bound on $\norm{\dot o^{(j)}}<v_o,\forall j \in [1;n_o]$ according to Definition \ref{def: obs}, we get the following modified upper bound on the radius dynamics:
  \begin{align}
      |\dot d^{(j)}(t)|\leq \max\Big(k_1,\sqrt{k_2^2+k_3^2}\Big)+v_o.\label{eqn:rad_bound2}
  \end{align}
Next, in order to find the upper bound on $|\dot\rad(t)|$, we first apply the norm inequality, and then we use the bound derived in Equations \eqref{eqn:rad_bound1} and \eqref{eqn:rad_bound2} to get the final upper bound on $|\dot \rad(t)|$ as follows:
\begin{align*}
     |\dot \rad(t)|&\leq \max\Big(k_1,\sqrt{k_2^2+k_3^2}\Big)+v_o.
\end{align*}
\section{Controller Synthesis}
We now derive a closed-form model-free control law to constrain the system state within the STT. It is important to note that, unlike traditional STT formulations \cite{das2024prescribed}, which do not consider the input constraint, this work considers the input constraint as defined in System Definition  \ref{subsec:sys}. Considering input constraint leads to a trade-off between performance and available actuation. In order to balance this trade-off, we will also derive the feasibility conditions.

We will derive the control law following a two-step procedure inspired by a backstepping-like design approach, similar to that described in \cite{PPCfeedback} and \cite{hard_soft}. At the first stage, we design a virtual velocity input to enforce the state of the system inside the STT and at the second stage, we design an acceleration level control to ensure that the virtual velocity input is tracked and thus the state remains inside the STT. Next we rewrite the EL system \eqref{eqn:EL} as follows:
\begin{subequations}\label{eq:EL_simp}
\begin{align}
    \dot{x} &= v, \label{eq:EL_simpa} \\
    \dot{v} &= V_M(x,v) + M(x)^{-1}\tau + M(x)^{-1}d, \label{eq:EL_simpb}
\end{align}
\end{subequations}
where $V_M(x,v)=-M(x)^{-1}(V(x,v)+G(x))$.
\subsection{Control Design}
The steps of the control design are as follows:

\textbf{Stage 1}: Given the STT $\Gamma(t)$ in \eqref{eqn:stt_ball} at first we define the normalized error $e_1(x,t)$ as follows:
\begin{align}
    e_1(x,t)=\frac{\norm{x(t)-\cen(t)}}{\rad(t)}.\notag
\end{align}
The intermediate control input $v_r(x_1,t)$ is given by:
\begin{align}
   v_r(x_1,t)=-\bar v\Psi(\e_1(x_1,t))\hat e_1(x_1,t),\label{eqn:r2}
\end{align}
where $\hat e_1(x_1,t)=\frac{(x(t)-\cen(t))}{\|x(t)-\cen(t)\|}$, $\bar v\in \R^+$ and the map $\Psi$ is the bounded transformation function $\Psi:\R^n\rightarrow\R^n$ is a continuously differentiable map which ensures that the control remains bounded within the control limits while maintaining the desired behavior. A detailed discussion has been given in Appendix \ref{a1:bounded}.

\textbf{Stage 2:} To ensure $v$ tracks the reference signal $v_r$ from Stage 1, at first we define the velocity tracking error as: $e_v(t)=v(t)-v_r(t)$. We enforce that $e_v(t)$ to be always remain bounded within exponentially decaying funnel constraint $\rho_v:\R_0^+\rightarrow \R^n$ given by:
\begin{align}\label{eqn:funnel}
    \rho_v(t)=\ex^{-\mu_vt}(p_v-q_v)+q_v,
\end{align}
where $p_v\in \R^n$ is the initial funnel width, $q_v\in\R^n$ is the steady state limit with $0_{n\times1}<q_v<p_v$ and $\mu_v\in\R^{n \times n}$ is a diagonal matrix that determines the decay rate of the funnel. Next, in order to enforce the funnel constrain \ie $-\rho_v(t)<e_v(t)<\rho_v(t)$, we proposed the following acceleration level control $\tau(t)$ given by:
\begin{align}
    \tau(t)=-diag(\Psi(\varepsilon_v))\bar \tau ,
\end{align}
where $\bar \tau\in \R^n$ is the maximum allowable torque and $\Psi(x)$ is the transformation function (Appendix \ref{a1:bounded}) and $\varepsilon_v(t)$ is the normalized error defined as: $\varepsilon_v=diag(\rho_v)^{-1}e_v(t).$
\begin{lemma}\label{lemm:a_r}
    A uniform bound on the reference acceleration $a_r$ or the upper bound on $\|\dot v_r\|$ is given by:
    \begin{align}
        \| \dot v_r\|\leq a_r,\notag
    \end{align}
    where $a_r=\alpha\frac{\bar v+p_v+\bar \cen}{\rad_{min}}+\beta\frac{\bar \rad}{\rad_{min}}+\theta\frac{\bar v+p_v+\bar \cen}{\rad_{min}}$ with $\alpha,\beta\text{ and }\theta$ can be found according to property of bounded transformation discussed in Propositions \ref{prop:a1}.
    \begin{proof}
        We first compute the time derivative of reference velocity $v_r(t)$ in Equations \eqref{eqn:r2}:
        \begin{align}
            \dot v_r(t) = -\bar v \dot \Psi(e_1(x_1,t))\hat e_1(x_1,t)-\bar v\Psi(e_1(x,t))\dot {\hat e}_1\label{eqn:dot_v_r}
        \end{align}
       In the first part of the proof, we compute the time derivative of $ \dot \Psi(e_1(x_1,t))$ as follows:
        \begin{align}
             \dot \Psi(e_1(x,t))&=\frac{\partial\Psi(e_1(x,t))}{\partial e_1}\dot e_1(x,t)\notag\\
             \dot \Psi(e_1(x,t))&=\frac{\partial\Psi(e_1(x,t))}{\partial e_1}\frac{(x-\cen)^\top(\dot x-\dot \cen)\rad-\norm{x-\cen}^2 \dot \rad}{\rad^2\|x-\cen\|},\notag
        \end{align}
        and then we use the norm inequality to get the following upper bound:
        \begin{align}
            | \dot \Psi(e_1(x,t))|\leq \Big|\frac{\partial\Psi(e_1(x,t))}{\partial e_1}\Big| \frac{\norm{\dot x\|+\|\dot \cen}}{\rad_{min}}+\Big|\frac{\partial\Psi(e_1(x,t))}{\partial e_1}e_1\Big|\frac{|\dot \rad|}{|\rad_{min}|}.\label{eqn: psi_bound}
        \end{align}
        Next, we use $\dot x=e_1+v_r$, along with the property of bounded transformation function $\Psi(e_1(x,t))$ in Proposition \ref{prop:a1} to upper bound $\Big|\frac{\partial\Psi(e_1(x,t))}{d e_1}\Big|$, $\Big|\frac{\partial\Psi(e_1(x,t))}{d e_1}e_1\Big|$ and consider $\alpha, \beta\in \R^+$ be the respective bound and  substitute these in equation \eqref{eqn: psi_bound} to get the final bound for $| \dot \Psi(e_1(x,t))|$ as follows:
        \begin{align}
            | \dot \Psi(e_1(x,t))|\leq\alpha\frac{\bar v+\norm{p_v}+\bar \cen}{\rad_{min}}+\beta\frac{\bar \rad}{\rad_{min}}.\label{eqn:bound_first}
        \end{align}
        In second stage we try to find  the time derivative $\dot {\hat e}_1$ as follows:
        \begin{align}
            \dot {\hat e}_1=\frac{1}{e_1\rad}\Big(\mathbf{I}-\hat e_1 \hat e_1^\top\Big)(\dot x(t)-\dot \cen(t)),\notag
        \end{align}
        which can be upper bounded using norm inequality and using upper bound on $\rad$,$\|\dot x\|$ and $\|\dot \cen\|$ as follows:
        \begin{align}
             \|\dot {\hat e}_1\|\leq\frac{\bar v+\norm{p_v}+\bar \cen}{e_1\rad_{min}}\label{eqn:hat_e}
        \end{align}
        Finally, we substitute the upper bounds obtained from Equations \eqref{eqn:bound_first} and \eqref{eqn:hat_e} into Equation \eqref{eqn:dot_v_r} to obtain the following bound:
        \begin{align}
            \|\dot v_r\|\leq\alpha\frac{\bar v+\norm{p_v}+\bar \cen}{\rad_{min}}+\beta\frac{\bar \rad}{\rad_{min}}+\frac{\Psi(e_1)}{e_1}\frac{\bar v+\norm{p_v}+\bar \cen}{\rad_{min}},\notag
        \end{align}
        we use upper bound on $|\frac{\Psi(e_1)}{e_1}|<\theta$ from Proposition \ref{prop:a1} and get the following upper bound:
        \begin{align}
             \|\dot v_r\|\leq\alpha\frac{\bar v+\norm{p_v}+\bar \cen}{\rad_{min}}+\beta\frac{\bar \rad}{\rad_{min}}+\theta\frac{\bar v+\norm{p_v}+\bar \cen}{\rad_{min}}.\notag
        \end{align}
    \end{proof}
\end{lemma}
\subsection{Feasibility Condition}
Since the proposed framework explicitly accounts for input constraints, we derive a feasibility condition that establishes a relationship among the system uncertainty, the upper bounds on the tube center and radius dynamics, and the available control authority. Satisfaction of this condition guarantees that the system state remains within the spatiotemporal tube $\Gamma(t)$ at all times, thereby ensuring safety. So given the EL system in \eqref{eqn:EL} under Assumptions \ref{ass:el1}, \ref{ass:el2} and \ref{ass:el3}, the STT $\Gamma(t)$, the funnel constraint and the input bound $\bar\tau$, the maximum allowable velocity $\bar v$ and torque $\bar \tau$ should follow the following conditions:
\begin{align}
    \bar  v\geq&\norm{p_v}+\bar \cen+\bar \rad\label{eqn:v_bar_feas}\\
    \bar \tau \geq & \frac{1}{\underline{m}}(V_M^{max}+\underline{m}_i\bar d+ \mu_v(p_v-q_v)+a_r\mathbf{e}_n)\label{eqn:tau_bar_feas},
\end{align}
where $\mathbf{e}_n=[1,\ldots,1]^\top\in\mathbb{R}^n$, $\bar{\cen}$ and $\bar{\rad}$ are the bounds on the STT center and radius dynamics, respectively,  given in Lemma~\ref{lem:STT_bound}, $V_M^{max}$ is the bound given by Assumption \ref{ass:el3}, and $a_r=\alpha\frac{\bar v+\norm{p_v}+\bar \cen}{\rad_{min}}+\beta\frac{\bar \rad}{\rad_{min}}+\theta\frac{\bar v+\norm{p_v}+\bar \cen}{\rad_{min}}$ as given in Lemma~\ref{lemm:a_r}.
\begin{remark}
The feasibility conditions in \eqref{eqn:v_bar_feas} can be verified offline and do not depend on the specifications of the FT-RAS task. Given the input constraints and the bounds provided by Assumptions~\ref{ass:el1}--\ref{ass:el3}, the design parameters can be selected to satisfy the feasibility conditions.
\end{remark}
Next, we present the theorem that formally summarizes the control law proposed in the work. 
\begin{theorem}
    Consider the EL system $\Sy$ in \eqref{eqn:EL} satisfying the Assumptions in \ref{ass:el1}- \ref{ass:el3}, a finite-time reach-avoid-stay (FT-RAS) specification as defined in Definition \ref{def: tras}, and the corresponding STT as derived in \eqref{eqn:stt_ball}.
    If the initial state  lies inside the tube $\ie$ $x(0)\in \Gamma(0)$ and the feasibility conditions in \eqref{eqn:v_bar_feas} is satisfied, then the closed form control laws
    \begin{align}
        v_r(t)=&-\bar v\Psi(\e_1(x_1,t))\hat e_1(x_1,t)\notag\\
         \tau(t)=&-diag(\Psi(\varepsilon_v))\bar \tau \label{eqn:control_law}
    \end{align}
    ensure that the system state $x(t)$ remains within the STT:
    \begin{align}
        x(t)\in \Gamma(t), \forall t\in \R_0^+.\notag
    \end{align}
    thereby satisfying the assigned  TRAS specification while respecting the input constraint $\ie$ $|\tau(t)|\leq \bar \tau, \forall t \in \R_0^+$.
    \begin{proof}
         We will prove the correctness of the control law in \eqref{eqn:control_law} in two stages.
         
         \textbf{Stage 1:} In this stage, we show that the velocity command $v_r(t)$ defined in \ref{eqn:r2} ensures that the system state $x(t)$ always remains within the STT. In order to show this, we proceed by method of contradiction. Suppose $\exists t_x\in \R^+$ such that $t_x$ is the first time instance when the state $x(t)$ is outside STT $\Gamma(t)$ on application of the velocity input $v_r(t)$. Then, at time $t_x$ as $e_1\rightarrow1$ the following condition should hold:
         \begin{align}
             \lim_{e_1\rightarrow1} \dot{e}_1>0.\label{eqn:cond}
         \end{align}
         Next, we compute error dynamics:
         \begin{align*}
             \dot e_1(x_1,t)=\frac{(x-\cen)^\top(\dot x-\dot c)}{\|x-\cen\|\rad}-\frac{e_1 \dot \rad}{\rad},
         \end{align*}
         and the state dynamics which can be rewritten as $\dot x(t)=v_r(t)+e_v(t)$ with $e_v(t)\leq p_v$ (we show in Stage 2). We substitute this error dynamics and state dynamics in \eqref{eqn:cond} to get following condition:
         \begin{align*}
             \lim_{e_1\rightarrow1}(x-\cen)^\top(\dot x-\dot c)>& e_1\|x-c\|\dot r\notag\\
             \implies  \lim_{e_1\rightarrow1}(x-\cen)^\top v_r>&(x-\cen)^\top\dot \cen+e_1\|x-\cen\|\dot r-(x-\cen)^\top e_v.\notag
         \end{align*}
         Then, we substitute $v_r(t)$ when $e_1 \rightarrow 1$ from \eqref{eqn:r2} in above equations:
         \begin{align*}
             \lim_{e_1\rightarrow1}-\|x-\cen\|\bar v>(x-\cen)^\top\dot \cen+e_1\|x-\cen\|\dot r-(x-\cen)^\top e_v\notag\\
             \implies \lim_{e_1\rightarrow1}\bar v\leq -\frac{(x-\cen)^\top \dot\cen}{\|x-\cen\|}-e_1\dot r+\frac{(x-\cen)^\top e_v}{\|x-\cen\|},
         \end{align*}
        and  we  apply norm inequality, boundary condition $e_1 \rightarrow 1$ and upper bound on $\norm{\dot \cen},|\dot r\|$ and $\|e_v\|<\|p_v\|$ (from Stage 2) to get the final condition which should be satisfied if any $t_x$ exists satisfying the condition in \eqref{eqn:cond}:
        \begin{align}
            \bar v<\bar \cen+ \bar \rad+\|p_v\|,
        \end{align}
         but this contradicts the feasibility condition in \eqref{eqn:v_bar_feas}. Thus, $x(t)\in \Gamma(t), \forall t \in \R_0^+$.

        \textbf{Stage 2}: Next, we show that the control input $\tau(t)$ defined in \eqref{eqn:control_law} ensures the velocity error $e_v(t)$ remains inside the funnel for all time $t\in\R_0^+$.

        Similar to Stage 1, we proceed by contradiction. Let $t_x\in \R^+$ be the first time instance when the velocity error $e_v(t)$ under the control law in \eqref{eqn:control_law}  violates the funnel constraint:
        \begin{align}
            \exists i \in [1;n],e_{v,i}(t_x)\leq-\rho_{v,i}(t_x)\text{ or }e_{v,i}(t_x)\geq\rho_{v,i}(t_x).\notag
        \end{align}
        Then, for time $t\in[0,t_x)$, the velocity error $e_v(t)$ should satisfy:
        \begin{align}
            -\rho_{v,i}<e_{v,i}<\rho_{v,i},\forall (t,i)\in[0,t_x)\times[1;n].\label{eqn:proof_1}
        \end{align}
        Next, we examine the two cases, first when the error approaches the upper funnel constraint and the second when the velocity error $e_v$ approaches the lower funnel constraint.

        \textbf{Case 1:} When there exists $i\in[1;n]$ such that $e_{v,i}(t)$ approaches the upper funnel constraint $\ie$, $e_{v,i}(t)\rightarrow\rho_{v,i}(t)$, which also implies that $e_{v,i}(t)-\rho_{v,i}(t)\rightarrow 0$. So if $\exists t_x$ then following \eqref{eqn:proof_1}, the velocity error must satisfy the following conditions:
        \begin{align}
            &(e_{v,i}(t)-\rho_{v,i}(t))\uparrow0\implies\lim_{e_{v,i}(t)-\rho_{v,i}(t))\uparrow0}\frac{d(e_{v,i}(t)-\rho_{v,i}(t))}{dt}>0\notag\\
            &\implies \lim_{e_{v,i}(t)-\rho_{v,i}(t))\uparrow0}\dot e_{v,i}(t)>\lim_{e_{v,i}(t)-\rho_{v,i}(t))\uparrow0} \dot \rho_{v,i}(t)\label{eqn:proof_2}
        \end{align}
        We substitute the derivative of funnel boundary $\dot \rho_{v,i}=-\mu_{v,i}\ex^{-\mu_{v,i}t}(p_{v,i}-q_{v,i})$ and also $\dot e_{v,i}=\dot v_i(t)-\dot v_r(t)$ in \eqref{eqn:proof_2}:
        \begin{align}
            \lim_{e_{v,i}(t)-\rho_{v,i}(t))\uparrow0} \dot v_i(t)> -\mu_{v,i}\ex^{-\mu_{v,i}t}(p_{v,i}-q_{v,i})-a_r.\label{eqn:proof_dotv}
        \end{align}
        Next, we look at the control input vector $\tau(t)=[\tau_1(t),\ldots,\tau_n(t)]^\top,\forall i\in[1;n]$:
        \begin{align}
            \lim_{e_{v,i}(t)-\rho_{v,i}(t))\uparrow0} \tau_i(t)=-\bar\tau_i,
        \end{align}
        since $\lim_{e_{v,i}(t)-\rho_{v,i}(t))\uparrow0} \varepsilon_{v,i}=1$, by construction of transformation function (Appendix \ref{a1:bounded}).

        Given the dynamics of $\dot v$ in \eqref{eq:EL_simpb}, we use the bound given by Assumption \ref{ass:el1}-\ref{ass:el3} to get the following upper bound:
        \begin{align}
             \lim_{e_{v,i}(t)-\rho_{v,i}(t))\uparrow0} \dot v_i(t)\leq \overline V_{M, i}-\underline{m}\bar \tau_i+\underline{m}_i\bar d_i.
        \end{align}
        Next, we use the feasibility condition from \eqref{eqn:tau_bar_feas} and substitute it in the above inequality to get the following conditions:
        \begin{align}
             \lim_{e_{v,i}(t)-\rho_{v,i}(t))\uparrow0} \dot v_i(t)\leq -\mu_{v,i}\ex^{-\mu_{v,i}t}(p_{v,i}-q_{v,i})-a_r
        \end{align}
        which contradicts the equation in \eqref{eqn:proof_dotv}. Hence, this shows that $e_{v,i}(t)\nrightarrow \rho_{v,i}(t),\forall (t,i)\in [0,t_x)\times [1;n],$ $\ie$, the velocity error $e_v(t)$ never approaches the upper funnel constraint over $t\in [0,t_x)$ in any discussion.

        \textbf{Case 2:} We use the same idea as in Case 1. At first we assume that there exists $i\in [1;n]$, such that $e_{v,i}$ approaches the lower funnel constraint, $\ie$, $e_{v,i}\rightarrow-\rho_{v,i}(t)$, which would implies that $e_{v,i}(t)+\rho_{v,i}(t)\rightarrow0$. Similar to Case 1, if $\exists t_x$, then following condition need to satisfied:
        \begin{align}
            \lim_{e_{v,i}(t)+\rho_{v,i}(t))\downarrow0}\frac{d(e_{v,i}(t)+\rho_{v,i}(t))}{dt}<0.
        \end{align}
        We simplify the above condition by substituting velocity error dynamics and the derivative of funnel boundary to get the following conditions:
        \begin{align}
                         \lim_{e_{v,i}(t)+\rho_{v,i}(t))\downarrow0}\dot v_i(t)<\mu_{v,i}\ex^{-\mu_{v,i}t}(p_{v,i}-q_{v,i})+a_r.\label{eqn:proof_case2}
        \end{align}
        Also when $e_{v,i}(t)+\rho_{v,i}(t))\downarrow0$ the control input in \eqref{eqn:control_law} $ \lim_{e_{v,i}(t)+\rho_{v,i}(t))\downarrow0} \tau_i(t)=\bar\tau_i$. Given the dynamics of $\dot v$ in \eqref{eq:EL_simpb}, we use the bound given by Assumption \ref{ass:el1}-\ref{ass:el3} to get the following lower bound:
      \begin{align}
             \lim_{e_{v,i}(t)+\rho_{v,i}(t))\downarrow0} \dot v_i(t)\geq \underline{V}_{M, i}+\underline{m}\bar \tau_i-\underline{m}_i\bar d_i.
        \end{align}
        Next, we use the feasibility condition from \eqref{eqn:tau_bar_feas} and substitute it in the above inequality to get the following conditions:
        \begin{align}
             \lim_{e_{v,i}(t)-\rho_{v,i}(t))\downarrow0} \dot v_i(t)\geq \mu_{v,i}\ex^{-\mu_{v,i}t}(p_{v,i}-q_{v,i})+a_r
        \end{align}
        which contradicts the equation in \eqref{eqn:proof_case2}. Hence, this shows that $e_{v,i}(t)\nrightarrow -\rho_{v,i}(t),\forall (t,i)\in [0,t_x)\times [1;n],$ $\ie$, the velocity error $e_v(t)$ never approaches the lower funnel constraint over $t\in [0,t_x)$ in any dimension. 

        Thus, it can be concluded that there is no $t_x$ at which $e_{v,i}(t)$ violates the funnel constraint in any dimensions. Therefore, control law \eqref{eqn:control_law} ensures:
        \begin{align}
             -\rho_{v,i}<e_{v,i}<\rho_{v,i},\forall (t,i)\in\R_0^+\times[1;n].\notag
        \end{align}

        Hence, the control law proposed in \eqref{eqn:control_law}, under feasibility conditions \eqref{eqn:v_bar_feas}-\eqref{eqn:tau_bar_feas}, ensures that $x(t)\in \Gamma(t),\forall t \in \R_0^+$, thus satisfying the FT-RAS task while respecting the input constraint $\ie$ $\tau(t)<\bar \tau,\forall t \in \R_0^+$.
    
    \end{proof}
\end{theorem}
\section{Case Studies}
We validate the proposed real-time STT framework through three case studies, a 2D mobile robot, a 3D quadrotor, and a satellite system supported by simulations and real-world mobile robot experiments.
\subsection{Mobile Robot}
We consider a differential drive robot operating in a 2D environment with dynamics adapted from \cite{dynamics_differential}, which selects a virtual control input on the robot that transforms the underactuated dynamics of differential drive systems into a fully actuated system. Further, we consider an input constraint of $\bar{\tau}=[13,13]^\top$N, chosen based on the specifications of the AgileX LiMO robot with mass as $4$kg. To test the robustness, we also introduce a bounded disturbance. We conduct both hardware and simulation experiments with mobile robots.

\textbf{Simulation Results for 2-D:}
In the simulation, we consider $n_o=20$ dynamic unsafe sets with randomly generated initial positions, radii, and velocities. The system is assigned a FT-RAS task, starting from the initial set $\mathbf{S}=\B([0,-0.3]^\top,0.8)$ and required to reach the target set $\mathbf{T}=\B([3,3]^\top,0.8)$ within a finite time while respecting the control input constraints $\bar \tau=[13,13]^\top $ N. 

The velocity of each unsafe set is randomly selected while remaining bounded by $v_o=0.025\,\mathrm{m/s}$. Additionally, the system is subjected to an unknown disturbance bounded by $\bar d=0.1$. The STT center parameters in in \eqref{eqn:cen_dyn}  are chosen as $k_1=0.06$, $k_2=0.03$, and $k_3=0.03$ to satisfy the condition $k_2>v_o$. Furthermore, the tube parameters are set to $\rad_a=0.5\,\mathrm{m}$, $\rad_{max}=0.45\,\mathrm{m}$, and $\rad_{min}=0.4\,\mathrm{m}$. For the second-stage controller, the time-varying decaying function parameters are selected as $\mu_v=0.5$, $p_v=0.05\e_n$, and $q_v=0.005\e_n$. Furthermore, the transformation function is chosen as given in Appendix~\ref{a1:bounded}. With these design parameters, the resulting feasibility bound on the control input in \eqref{eqn:v_bar_feas}-\eqref{eqn:tau_bar_feas}  satisfies the required feasibility condition.

 Figure~\ref{fig:plot_2d} (a)-(b) presents snapshots of the system evolution at $t=12.5\,\mathrm{s}$ and $t=56.0\,\mathrm{s}$. The third plot Figure~\ref{fig:plot_2d} (c) illustrates the evolution of the spatiotemporal tube with time along with system trajectory inside it, showing the system is always inside the tube and successfully reaches the target set at $t=187\,\mathrm{s}$. The final subplot Figure~\ref{fig:plot_2d} (d) depicts the control input profile throughout the simulation, demonstrating that the applied inputs remain well within the allowable bounds. Full video of simulation is available at \href{https://youtu.be/5QB1i58U29E}{link}.

\begin{figure}
    \centering
    \includegraphics[width=0.9\linewidth]{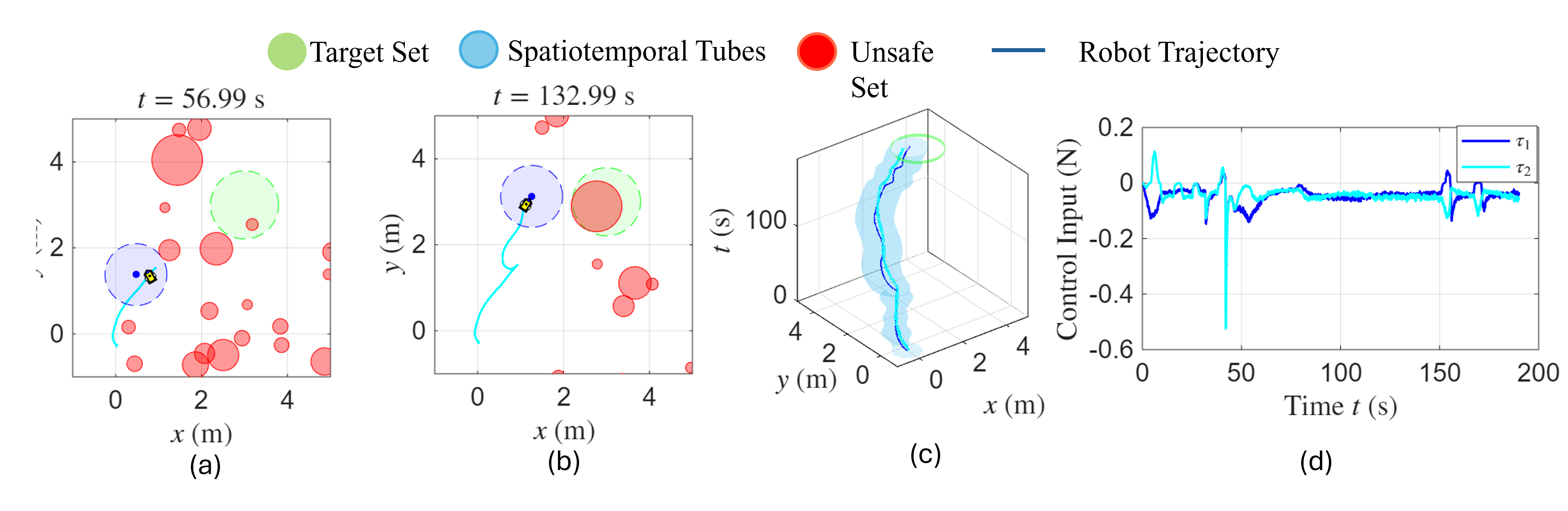}
    \caption{Simulation results for 2-D Mobile Robot, (a) and (b) present snapshots of the robot trajectory at two different time instants, (c) presents the evolution of the STT over time, showing that the system trajectory remains entirely within the STT throughout the motion, and (d) presents the control inputs over time, with the input bound given by $\bar{\tau}=[13,13]^\top$ N.}
    \label{fig:plot_2d}
\end{figure}
\textbf{Hardware Results for 2-D:} In the hardware experiments, we consider an AgileX differential-drive robot with the input constraint set to $\bar{\tau}=[13,13]^\top$. We consider the target set $\mathbf{T}=\B([4.12,2.78]^\top,0.5)$. The entire experimental setup is implemented using the Robot Operating System (ROS), which enables real-time communication between the sensing, processing, and control modules. A motion capture system provides real-time state feedback for the robot. As shown in Figure~\ref{fig:hardware_1}, a real-time projection system displays the position of the dynamic unsafe set, while the corresponding unsafe-set information is simultaneously published to a ROS topic. The mobile robot subscribes to this topic to obtain real-time information about the environment and uses it to compute safe control actions during navigation.

Figure~\ref{fig:hardware_1} presents the results of the hardware experiments. The first two sub-figures show snapshots of the experiment at two different time instants, demonstrating that the mobile robot maintains a safe distance from the dynamic unsafe set. The third sub-figure illustrates the evolution of the STT over time, together with the robot trajectory (shown in blue), which remains entirely within the STT throughout the experiment. The final sub-figure presents the control inputs over time, confirming that they always remain below the prescribed maximum allowable input. The system reaches the target set at $t=192$s, full video is available at \href{https://youtu.be/5QB1i58U29E}{link}.

\begin{figure}
    \centering
    \includegraphics[width=0.9\linewidth]{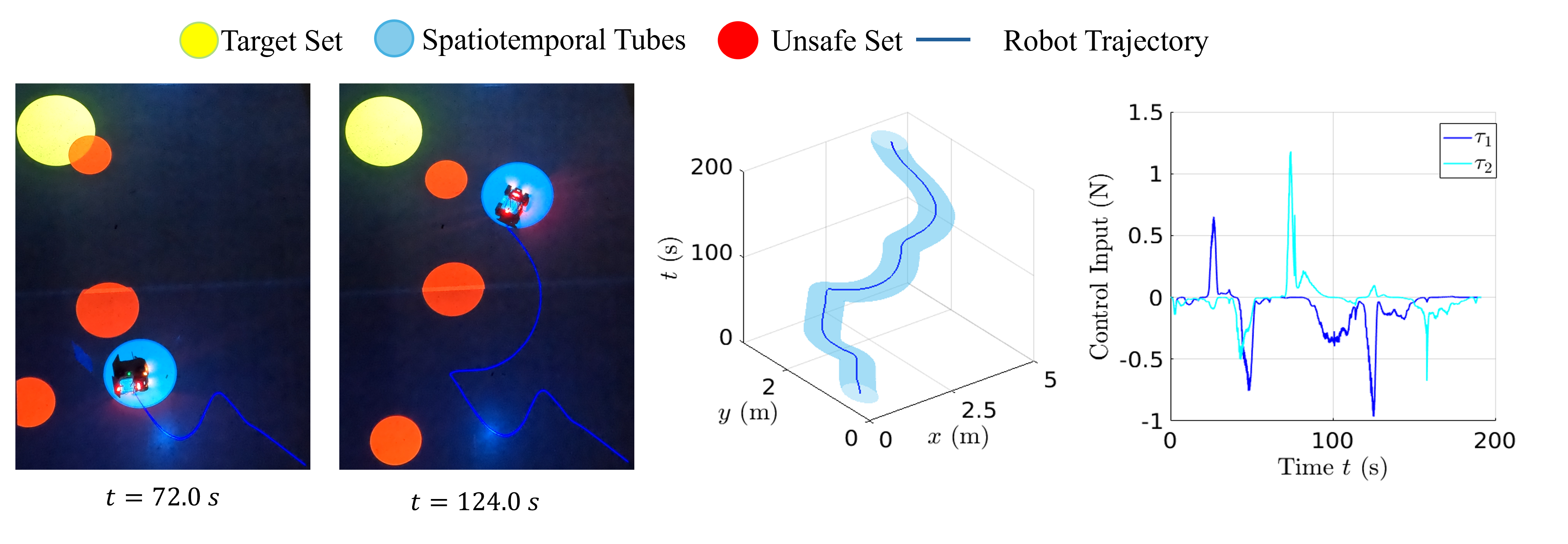}
    \caption{Hardware results for the 2D mobile robot. The first two plots depict snapshots of the hardware experiment at two distinct time instants. The third plot illustrates the evolution of the spatiotemporal tube (STT) over time, while the fourth plot shows the control input. The actuator input is constrained by $(\bar{\tau}=[10,10]^\top$ N). \href{https://youtu.be/5QB1i58U29E}{Video}
}
    \label{fig:hardware_1}
\end{figure} 
\subsection{Quadrotor}
We further evaluate the proposed framework on a quadrotor operating in a three-dimensional environment with dynamics adapted from \cite{APF_drone}. The system starts from $\mathbf{S}=\B([1,1,1]^\top,1)$ and is required to reach the target set $\mathbf{T}=\B([8,8,8]^\top,1)$ within a finite time while respecting the control input constraint $\bar{\tau}=[10,10,10]^\top~\text{N}$. The STT center parameters in \eqref{eqn:cen_dyn} are chosen as, $k_1=0.4$ and $k_2^{(j)}=k_3^{(j)}=1$ to satisfy the condition $k_2>v_o$. The acting radius is chosen as $\rad_a=0.6~\text{m}$, while the maximum and minimum allowable tube radii are set to $\rad_{\max}=0.54~\text{m}$ and $\rad_{\min}=0.42~\text{m}$, respectively. We consider $n_o=4$ dynamic unsafe sets with randomly generated radii and velocities satisfying the upper bound $v_o=0.1~\text{m/s}$.

\begin{figure}
    \centering
    \includegraphics[width=0.9\linewidth]{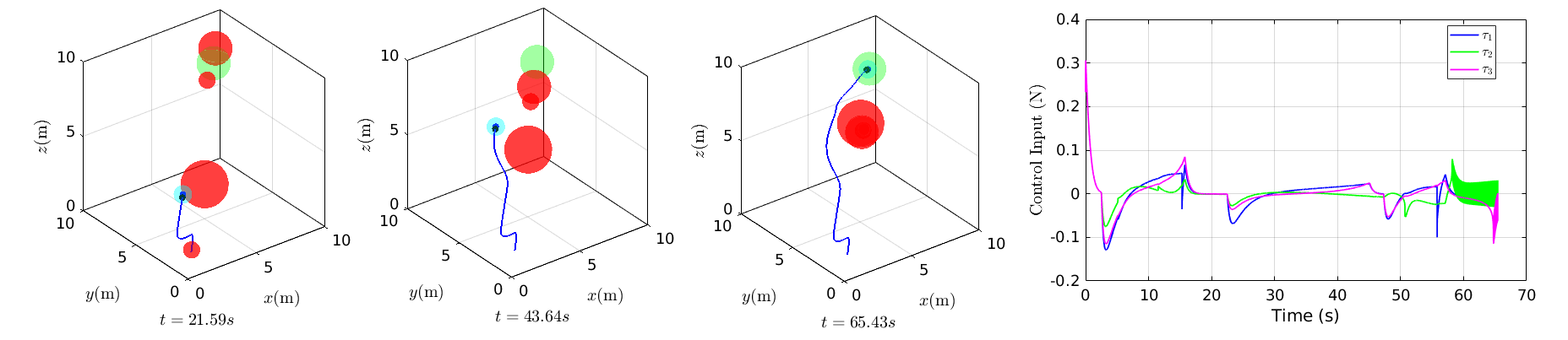}
    \caption{Simulation results for the 3-D quadrotor. The first two plots present snapshots of the STT and the system trajectory at two different time instants. The third plot presents the control inputs over time under the input constraint $\bar{\tau}=[10,10,10]^\top$ N. \href{https://youtu.be/5QB1i58U29E}{Video}}
    \label{fig:uav_3d}
\end{figure}
For the second-stage controller, the time-varying decaying function parameters are selected as $\mu_v=0.1$, $p_v=0.1\e_n$, and $q_v=0.05\e_n$. Furthermore, the transformation function is chosen as given in Appendix~\ref{a1:bounded}. With these design parameters, the resulting feasibility bound \eqref{eqn:v_bar_feas}-\eqref{eqn:tau_bar_feas} on the control input is $[9.06,9.06,9.06]^\top~\text{N}$, which is well below the available control authority. Figure~\ref{fig:uav_3d} presents snapshots of the quadrotor trajectory at two time instants, and the system reaches the target set at $t=66.9$s. The final subplot in Figure~\ref{fig:uav_3d} illustrates corresponding control effort, demonstrating that the generated control input $\tau(t)$ remains well within the prescribed input constraints while simultaneously ensuring safety throughout the maneuver.  Full video of the simulation is available at \href{https://youtu.be/5QB1i58U29E}{link}.

\subsection{Rotating Spacecraft}

As a three-dimensional benchmark, we consider the attitude reorientation of a rigid spacecraft \cite{khalil2002nonlinear} while avoiding a forbidden pointing direction. Such constraints commonly arise in spacecraft operations, where optical payloads or star trackers must not point directly toward the Sun to prevent sensor saturation or permanent damage. Consequently, the spacecraft must reorient itself to the desired attitude while avoiding a set of unsafe orientations corresponding to Sun-pointing configurations.

The rotational dynamics are governed by Euler's rigid-body equations
\begin{align}
I_x\dot{\omega}_x-(I_y-I_z)\omega_y\omega_z &= \tau_x,\\
I_y\dot{\omega}_y-(I_z-I_x)\omega_z\omega_x &= \tau_y,\\
I_z\dot{\omega}_z-(I_x-I_y)\omega_x\omega_y &= \tau_z,
\end{align}
where $\omega=[\omega_x,\omega_y,\omega_z]^\top$ denotes the body angular velocity, $\tau=[\tau_x,\tau_y,\tau_z]^\top$ is the control torque, and $I_x$, $I_y$, and $I_z$ are the principal moments of inertia. The spacecraft attitude is represented using the Euler angles $q=[\phi,\theta,\psi]^\top$, whose kinematics are given by
$\omega = T(q)\dot q$,
where $T(q)$ denotes the standard Euler-angle kinematic transformation matrix.

The spacecraft is required to reorient from the initial attitude
$q(0)=[0,\,0,\,0]^\top$
to the desired attitude
$q_d=[2.5,\,2.0,\,1.5]^\top$
while avoiding a spherical unsafe region centered at
$q_u=[1.8,\,1.5,\,1.0]^\top$
with radius $0.4$ rad, representing the set of Sun-pointing orientations. The spacecraft inertia matrix is chosen as
$J=\mathrm{diag}(1.5,\,1.2,\,1.0)$, and the system input constraint is set to $\bar{\tau}=[0.23,\,0.23,\,0.23]^\top~\mathrm{N\cdot m}$.
The controller parameters are selected as
$k_1=0.4$, $k_2=0.2$, $k_3=0.2$, with an acting radius of $\rad_a=0.6$m and a maximum tube radius of $\rad_{\max}=0.54$ rad. With these design parameters, the resulting feasibility bound on the control input in \eqref{eqn:v_bar_feas}-\eqref{eqn:tau_bar_feas}  satisfies the required feasibility condition.

Figure ~\ref{fig:spacecraft} illustrates the spacecraft attitude evolution together with the synthesized STT in the attitude space. The generated tube safely steers the spacecraft around the forbidden attitude region before converging to the desired orientation at $t=202.19$s. Throughout the maneuver, the spacecraft remains inside the STT while respecting the input constraints, thereby guaranteeing satisfaction of the prescribed FT-RAS specification.

\begin{figure}
    \centering
    \includegraphics[width=1.0\linewidth]{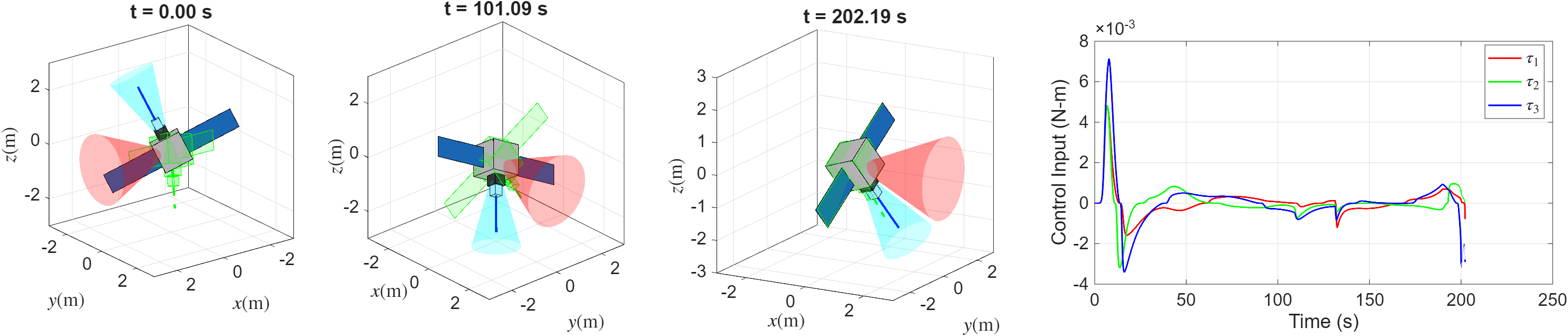}
    \caption{Simulation results for spacecraft attitude reorientation. The unsafe regions are represented by red cones, while the STT is depicted by the blue cone. The target spacecraft orientation is illustrated by the green wireframe spacecraft. The first two plots show snapshots of the spacecraft attitude at two different time instants. The third plot presents the control inputs over time under the input constraint $\bar{\tau}=[0.23,0.23,0.23]^\top$ N. \href{https://youtu.be/5QB1i58U29E}{Video}}
    \label{fig:spacecraft}
\end{figure}
\section{Comparison}
To evaluate the effectiveness of the proposed STT-based framework to solve the FT-RAS task with input constraints, we present a comparison with the existing real-time STT framework \cite{STT_Real} and also present an extensive qualitative comparison.
\subsection{Comparison with existing STT framework}

\begin{table*}[t]
\centering
\begin{adjustbox}{max width=\textwidth}
\begin{threeparttable}
\caption{Comparing proposed spatiotemporal tubes with classical algorithms}
\begin{tabular}{lcccccccccccc}
\hline
\textbf{Algorithm} & \multicolumn{2}{c}{\textbf{\begin{tabular}[c]{@{}l@{}}Closed-form \\ Control \end{tabular}}} & \multicolumn{2}{c}{\textbf{\begin{tabular}[c]{@{}l@{}}Formal \\ Guarantee \end{tabular}}} & \multicolumn{2}{c}{\textbf{\begin{tabular}[c]{@{}l@{}}Unknown \\ Dynamics \end{tabular}}} & \multicolumn{2}{c}{\textbf{\begin{tabular}[c]{@{}l@{}}Input\\ Constraint \end{tabular}}} & \multicolumn{2}{c}{\textbf{\begin{tabular}[c]{@{}l@{}}Bounded \\ Disturbance \end{tabular}}} & \multicolumn{2}{c}{\textbf{\begin{tabular}[c]{@{}l@{}}Dynamic \\ Environment \end{tabular}}} \\ \hline
ILQR based Planning \cite{ILQR} & \multicolumn{2}{c}{\xmark} & \multicolumn{2}{c}{\xmark} & \multicolumn{2}{c}{\xmark} & \multicolumn{2}{c}{\cmark} & \multicolumn{2}{c}{\xmark} & \multicolumn{2}{c}{\cmark} \\
NMPC \cite{NMPC} & \multicolumn{2}{c}{\xmark} & \multicolumn{2}{c}{\xmark} & \multicolumn{2}{c}{\xmark} & \multicolumn{2}{c}{\cmark} & \multicolumn{2}{c}{\xmark} & \multicolumn{2}{c}{\cmark} \\ 
Correct by Design CBF\cite{EL_CBF} & \multicolumn{2}{c}{\xmark} & \multicolumn{2}{c}{\cmark} & \multicolumn{2}{c}{\xmark} & \multicolumn{2}{c}{\cmark} & \multicolumn{2}{c}{\xmark} & \multicolumn{2}{c}{\cmark} \\
Neural CBF \cite{Neural_CBF} & \multicolumn{2}{c}{\xmark} & \multicolumn{2}{c}{\cmark} & \multicolumn{2}{c}{\xmark} & \multicolumn{2}{c}{\cmark} & \multicolumn{2}{c}{\xmark} & \multicolumn{2}{c}{\xmark} \\
Symbolic Control \cite{SCOTS}& \multicolumn{2}{c}{\xmark} & \multicolumn{2}{c}{\cmark} & \multicolumn{2}{c}{\xmark} & \multicolumn{2}{c}{\cmark} & \multicolumn{2}{c}{\cmark} & \multicolumn{2}{c}{\xmark} \\
 Real Time STT \cite{STT_Real} & \multicolumn{2}{c}{\cmark} & \multicolumn{2}{c}{\cmark} & \multicolumn{2}{c}{\cmark} & \multicolumn{2}{c}{\xmark} & \multicolumn{2}{c}{\cmark} & \multicolumn{2}{c}{\cmark}\\
\textbf{Proposed} & \multicolumn{2}{c}{\cmark} & \multicolumn{2}{c}{\cmark} & \multicolumn{2}{c}{\cmark} & \multicolumn{2}{c}{\cmark} & \multicolumn{2}{c}{\cmark} & \multicolumn{2}{c}{\cmark} \\
\hline
\end{tabular}
\label{tab:comp}
  
\end{threeparttable}
\end{adjustbox}
\end{table*}

In this example, we consider a 2-D environment to compare the existing real-time STT framework in \cite{STT_Real} with the proposed input-constrained STT framework. The control input bound is chosen as $\bar{\tau}=[2.5,2.5]^\top N$, considering the mass of the system to be $1kg$, which is equivalent to the Turtlebot3. For the framework in \cite{STT_Real}, the prescribed reaching time is set to $t_c=200$s, while for the proposed approach, the design parameters are selected to satisfy the feasibility conditions.

Figure~\ref{fig:comp_1}(a) illustrates the system trajectory and corresponding control inputs obtained using the real-time STT framework of \cite{STT_Real}, whereas Figure~\ref{fig:comp_1}(b) presents the results for the proposed input-constrained STT framework. In both cases, the computational burden remains low since all components of the controller have closed-form expressions.

Although the framework in \cite{STT_Real} drives the system more quickly toward the target, it fails to satisfy the input constraints. In particular, the control inputs exceed the allowable bounds during the initial phase of the task due to an unsafe set moving directly toward the system. Another violation occurs around $t=84\,\mathrm{s}$, when the target region, previously obstructed by an unsafe set, becomes partially accessible. This causes a sudden increase in the target-driven component of the STT center dynamics, resulting in a noticeable control spike. Furthermore, even with a relatively large prescribed reaching time of $t_c=200\,\mathrm{s}$, the control effort may become large as $t$ approaches $t_c$, especially when slow-moving unsafe sets remain near the target. This behavior arises from the increasingly aggressive target-driven term required to satisfy the prescribed-time constraint.

In contrast, the proposed framework guarantees satisfaction of the input constraints throughout the execution while maintaining safety and ensuring the successful completion of the task within a finite time. The resulting trajectories remain feasible at all times, demonstrating the effectiveness of the proposed input-constrained STT design. Full comparison video is available at \href{https://youtu.be/5QB1i58U29E}{link}.

Thus, although system reaches the target more quickly in Figure~\ref{fig:comp_1}(a) when no input constraints are imposed. However, in practical applications, physical systems may not be capable of generating such large control inputs, which can lead to actuator saturation and potential safety violations. In contrast, the proposed method shown in Figure~\ref{fig:comp_1}(b) results in a slower system response while ensuring that the prescribed input constraints are always satisfied, thereby maintaining safety throughout the maneuver. This highlights the inherent trade-off between rapid task completion and strict adherence to input constraints.

\begin{figure}
    \centering
    \includegraphics[width=0.7\linewidth]{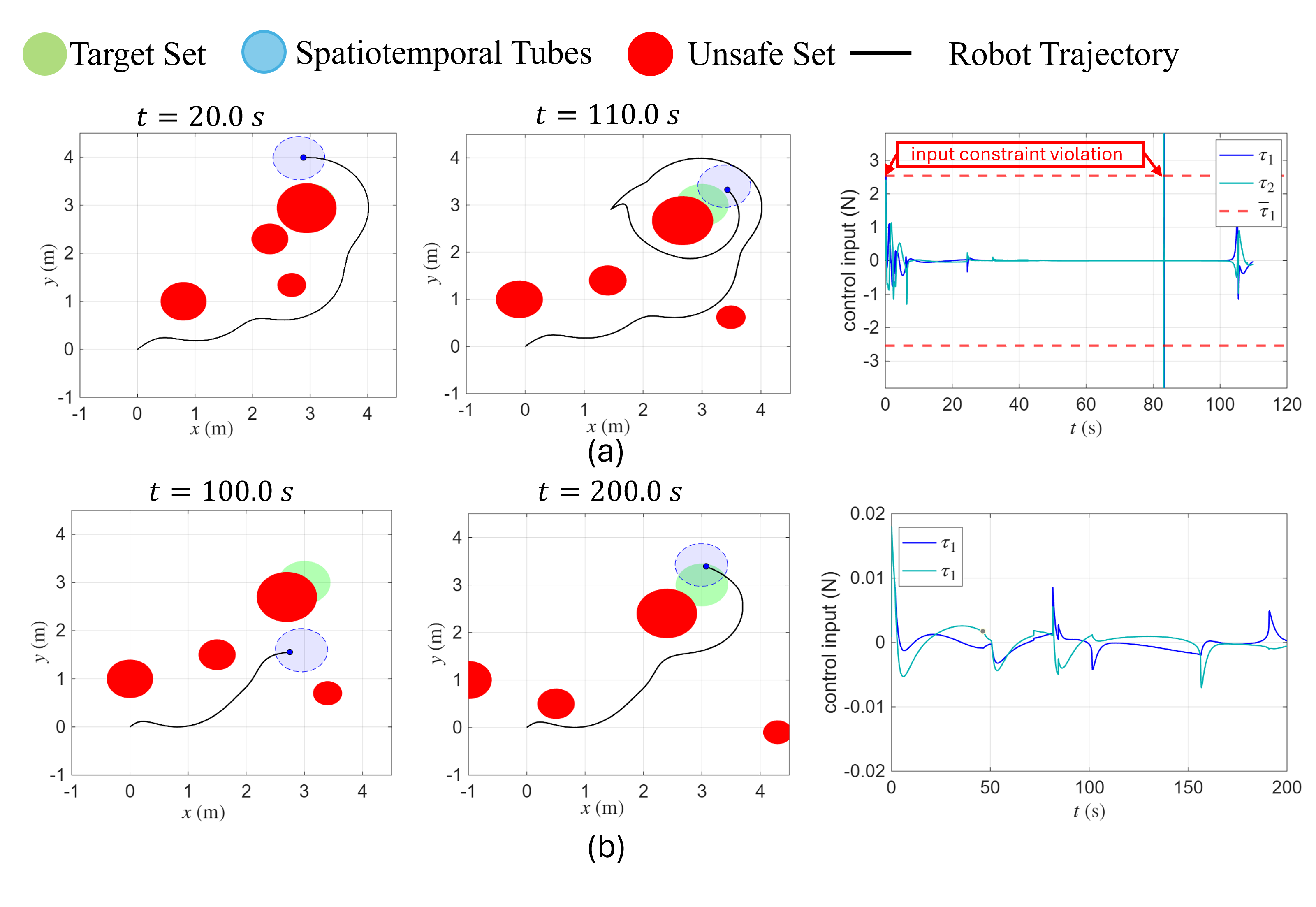}
    \caption{Comparison between (a) existing STT method \cite{STT_Real} and (b) proposed method with input constraint $\bar \tau=[2.5,2.5]^\top$ N. \href{https://youtu.be/5QB1i58U29E}{Video}} 
    \label{fig:comp_1}
\end{figure}
\subsection{Qualitative Comparison}
In this subsection, we provide a qualitative comparison of the proposed framework with existing approaches, as summarized in Table~\ref{tab:comp}. Specifically, the comparison includes ILQR-based planning \cite{ILQR}, NMPC \cite{NMPC}, Correct-by-Design CBF \cite{EL_CBF}, Neural CBF \cite{Neural_CBF}, Symbolic Control \cite{SCOTS}, and Real-Time STT \cite{STT_Real}. The methods are evaluated based on their computational efficiency, availability of formal guarantees, robustness to unknown dynamics and external disturbances, capability to accommodate input constraints, and suitability for real-time implementation.

\section{Conclusion and Future Work}
In this work, we proposed a real-time spatiotemporal tube (STT)-based framework to solve the reach-avoid problem for unknown Euler--Lagrange (EL) systems subjected to input constraints. Unlike existing approaches that often neglect input constraints or rely on accurate system models, the proposed framework is model-free, computationally efficient, and provides formal safety guarantees while explicitly accounting for actuator limitations. 
The effectiveness of the proposed approach is demonstrated through extensive simulation and hardware experiments on a 2D mobile robot, a 3D quadrotor, and a spacecraft system. 

Future work will focus on extending the proposed framework to more general classes of nonlinear systems, discrete-time systems, and stochastic systems while preserving the formal safety guarantees and input constraint satisfaction. Another important direction is to extend the framework to handle more general task specifications, such as Linear Temporal Logic (LTL), Signal Temporal Logic (STL), and Probabilistic Signal Temporal Logic (PrSTL). 


 \appendix
 \section*{Appendix}
 \section{Bounded Transformation}\label{a1:bounded}
The bounded transformation function $\Psi:\R^n\rightarrow\R^n$ is a smooth map that ensures the control inputs always remain within the limits. We defined the transformation function as  vector $\Psi(s)=[\Psi_i(s_1),\ldots,\Psi_n(s_n)]^\top$ with each of its components $\Psi_i(s_i)$ defined as follows:
\begin{align*}
    \Psi_i(s_i) =
\begin{cases}
-1, & s_i \in (-\infty, -1], \\
0,  & s_i = 0, \\
1,  & s_i \in [1, \infty),
\end{cases}
\end{align*}
and $\Psi_i(s_i)$ is non decreasing for all $s_i\in(-\infty,\infty)$. The next proposition lists some of the properties of the transformation function.
\begin{proposition}\label{prop:a1}
    The bounded transformation function $\Psi(s)$ have the following properties by its definition:
    \begin{enumerate}
        \item[i] There exists $\alpha\in \R^+$ satisfying the following conditions: $\|\frac{\partial \Psi(s)}{\partial s}\|\leq \alpha$.
        \item[ii] There exists $\beta\in \R^+$ satisfying the following conditions: $\|\frac{\partial\Psi(s)}{\partial s}s\|\leq \beta$.
        \item[iii]  There exists $\theta\in \R^+$ satisfying the following conditions: $\|\frac{\Psi(s)}{s}\|\leq \theta$.
    \end{enumerate}
\end{proposition}
An example of such a transformation function can be 
\[
\Psi_i(s_i)=
\begin{cases}
-1, & s_i \le -1, \\[6pt]
\displaystyle
2\,
\frac{
e^{-1/(s_i+1)^a}
}{
e^{-1/(s_i+1)^a}
+
e^{-1/(1-s_i)^a}
}
-1,
& -1 < s_i < 1, \\[12pt]
1, & s_i \ge 1.
\end{cases}
\]
\bibliographystyle{unsrt} 
\bibliography{sources} 

@article{tayal2026collision,
  title={A collision cone approach for control barrier functions},
  author={Tayal, Manan and Goswami, Bhavya Giri and Rajgopal, Karthik and Singh, Rajpal and Rao, M Tejas and Keshavan, Jishnu and Jagtap, Pushpak and Yadukumar, Shishir Nadubettu},
  journal={IEEE Transactions on Control Systems Technology},
  year={2026},
  publisher={IEEE}
}

@article{das2025spatiotemporal,
  title={Spatiotemporal tubes for temporal reach-avoid-stay tasks in unknown systems},
  author={Das, Ratnangshu and Basu, Ahan and Jagtap, Pushpak},
  journal={IEEE Transactions on Automatic Control},
  year={2025},
  publisher={IEEE}
}

@inproceedings{khaled2021omegathreads,
  title={Omegathreads: {S}ymbolic controller design for $\omega$-regular objectives},
  author={Khaled, Mahmoud and Zamani, Majid},
  booktitle={Proceedings of the 24th International Conference on Hybrid Systems: Computation and Control},
  pages={1--7},
  year={2021}
}

@inproceedings{sundarsingh2023scalable,
  title={Scalable distributed controller synthesis for multi-agent systems using barrier functions and symbolic control},
  author={Sundarsingh, David Smith and Bhagiya, Jay and Chatrola, Jeel and Saoud, Adnane and Jagtap, Pushpak and others},
  booktitle={62nd IEEE Conference on Decision and Control (CDC)},
  pages={6436--6441},
  year={2023}
}

@article{saoud2021compositional,
  title={Compositional abstraction-based synthesis for interconnected systems: {A}n approximate composition approach},
  author={Saoud, Adnane and Jagtap, Pushpak and Zamani, Majid and Girard, Antoine},
  journal={IEEE Transactions on Control of Network Systems},
  volume={8},
  number={2},
  pages={702--712},
  year={2021},
  publisher={IEEE}
}

@phdthesis{saoud2019compositional,
  title={Compositional and efficient controller synthesis for cyber-physical systems},
  author={Saoud, Adnane},
  year={2019},
  school={Universit{\'e} Paris Saclay (COmUE)}
}

@article{das2024prescribed,
  title={Prescribed-time reach-avoid-stay specifications for unknown systems: A spatiotemporal tubes approach},
  author={Das, Ratnangshu and Jagtap, Pushpak},
  journal={IEEE Control Systems Letters},
  year={2024},
  publisher={IEEE}
}

@book{ELbook,
  title={Euler-Lagrange systems},
  author={Ortega, Romeo and Loria, Antonio and Nicklasson, Per Johan and Sira-Ramirez, Hebertt},
  year={1998},
  publisher={Springer}
}

@book{ELbook2,
author = {Siciliano, Bruno and Sciavicco, Lorenzo and Villani, Luigi and Oriolo, Giuseppe},
title = {Robotics: Modelling, Planning and Control},
year = {2010},
isbn = {1849966346},
publisher = {Springer Publishing Company, Incorporated}
}

@article{ELbounds,
  title={Robot control by using only joint position measurements},
  author={Nicosia, S and Tomei, P},
  journal={IEEE Transactions on Automatic control},
  volume={35},
  number={9},
  pages={1058--1061},
  year={1990},
  publisher={IEEE}
}

@book{tabuada2009verification,
  title={Verification and control of hybrid systems: a symbolic approach},
  author={Tabuada, Paulo},
  year={2009},
  publisher={Springer Science \& Business Media}
}

@article{APF,
  title={Numerical potential field techniques for robot path planning},
  author={Barraquand, Jerome and Langlois, Bruno and Latombe, J-C},
  journal={IEEE transactions on systems, man, and cybernetics},
  volume={22},
  number={2},
  pages={224--241},
  year={1992},
  publisher={IEEE}
}

@ARTICLE{APF_drone,
  author={Pan, Zhenhua and Zhang, Chengxi and Xia, Yuanqing and Xiong, Hao and Shao, Xiaodong},
  journal={IEEE Transactions on Circuits and Systems II: Express Briefs}, 
  title={An Improved Artificial Potential Field Method for Path Planning and Formation Control of the Multi-{UAV} Systems}, 
  year={2022},
  volume={69},
  number={3},
  pages={1129-1133},
  doi={10.1109/TCSII.2021.3112787}}

@article{MPC,
  title={Path planning of collision avoidance for unmanned ground vehicles: {A} nonlinear model predictive control approach},
  author={Hang, Peng and Huang, Sunan and Chen, Xinbo and Tan, Kok Kiong},
  journal={Proceedings of the Institution of Mechanical Engineers, Part I: Journal of Systems and Control Engineering},
  volume={235},
  number={2},
  pages={222--236},
  year={2021},
  publisher={SAGE Publications Sage UK: London, England}
}

@inproceedings{SCOTS,
  title={{SCOTS: A} tool for the synthesis of symbolic controllers},
  author={Rungger, M. and Zamani, M.},
  booktitle={the 19th International Conference on Hybrid Systems: Computation and Control},
  pages={99--104},
  year={2016}
}

@ARTICLE{CBF,
  author={Ames, Aaron D. and Xu, Xiangru and Grizzle, Jessy W. and Tabuada, Paulo},
  journal={IEEE Transactions on Automatic Control}, 
  title={Control Barrier Function Based Quadratic Programs for Safety Critical Systems}, 
  year={2017},
  volume={62},
  number={8},
  pages={3861-3876},
  doi={10.1109/TAC.2016.2638961}}

@article{RRTs,
author = {Sertac Karaman and Emilio Frazzoli},
title ={Sampling-based algorithms for optimal motion planning},
journal = {The International Journal of Robotics Research},
volume = {30},
number = {7},
pages = {846-894},
year = {2011},
doi = {10.1177/0278364911406761},
URL = {https://doi.org/10.1177/0278364911406761},
eprint = {https://doi.org/10.1177/0278364911406761}
}

@article{PPCfeedback,
  title={A low-complexity global approximation-free control scheme with prescribed performance for unknown pure feedback systems},
  author={Bechlioulis, Charalampos P and Rovithakis, George A},
  journal={Automatica},
  volume={50},
  number={4},
  pages={1217--1226},
  year={2014},
  publisher={Elsevier}
}

@INPROCEEDINGS{hard_soft,
  author={Mehdifar, Farhad and Bechlioulis, Charalampos P. and Dimarogonas, Dimos V.},
  booktitle={IEEE 61st Conference on Decision and Control (CDC)}, 
  title={Funnel Control Under Hard and Soft Output Constraints}, 
  year={2022},
  volume={},
  number={},
  pages={4473-4478},
  doi={10.1109/CDC51059.2022.9992628}}

@book{khalil2002nonlinear,
  title={Nonlinear systems},
  author={Khalil, Hassan K and Grizzle, Jessy W},
  volume={3},
  year={2002},
  publisher={Prentice hall Upper Saddle River, NJ}
}

@article{bhat2000finite,
  title={Finite-time stability of continuous autonomous systems},
  author={Bhat, Sanjay P and Bernstein, Dennis S},
  journal={SIAM Journal on Control and optimization},
  volume={38},
  number={3},
  pages={751--766},
  year={2000},
  publisher={SIAM}
}

@article{VCZ,
  title={Scalable and Approximation-free Symbolic Control for Unknown Euler-Lagrange Systems},
  author={Das, Ratnangshu and Sawarkar, Shubham and Jagtap, Pushpak},
  journal={arXiv preprint arXiv:2509.19859},
  year={2025}
}

@ARTICLE{STT_Real,
  author={Das, Ratnangshu and Upadhyay, Siddhartha and Jagtap, Pushpak},
  journal={IEEE Robotics and Automation Letters}, 
  title={Real-Time Spatiotemporal Tubes for Dynamic Unsafe Sets}, 
  year={2026},
  volume={11},
  number={2},
  pages={2146-2153},
  doi={10.1109/LRA.2025.3645667}}

@ARTICLE{ILQR,
  author={Chen, Jianyu and Zhan, Wei and Tomizuka, Masayoshi},
  journal={IEEE Transactions on Intelligent Vehicles}, 
  title={Autonomous Driving Motion Planning With Constrained Iterative LQR}, 
  year={2019},
  volume={4},
  number={2},
  pages={244-254},
  doi={10.1109/TIV.2019.2904385}}

@INPROCEEDINGS{EL_CBF,
  author={Cortez, Wenceslao Shaw and Dimarogonas, Dimos V.},
  booktitle={2020 American Control Conference (ACC)}, 
  title={Correct-by-Design Control Barrier Functions for Euler-Lagrange Systems with Input Constraints}, 
  year={2020},
  volume={},
  number={},
  pages={950-955},
  doi={10.23919/ACC45564.2020.9147367}}

@article{NMPC,
  title={Nonlinear model predictive path following controller with obstacle avoidance},
  author={Sanchez, Ignacio and D’Jorge, Agustina and Raffo, Guilherme V and Gonzalez, Alejandro H and Ferramosca, Antonio},
  journal={Journal of Intelligent \& Robotic Systems},
  volume={102},
  number={1},
  pages={16},
  year={2021},
  publisher={Springer}
}

@inproceedings{Neural_CBF,
  title={Safe control under input limits with neural control barrier functions},
  author={Liu, Simin and Liu, Changliu and Dolan, John},
  booktitle={Conference on Robot Learning},
  pages={1970--1980},
  year={2023},
  organization={PMLR}
}

@article{Safety_critical,
  title={The safety filter: A unified view of safety-critical control in autonomous systems},
  author={Hsu, Kai-Chieh and Hu, Haimin and Fisac, Jaime F},
  journal={Annual Review of Control, Robotics, and Autonomous Systems},
  volume={7},
  year={2023},
  publisher={Annual Reviews}
}

@article{yin2024formal,
  title={Formal synthesis of controllers for safety-critical autonomous systems: Developments and challenges},
  author={Yin, Xiang and Gao, Bingzhao and Yu, Xiao},
  journal={Annual Reviews in Control},
  volume={57},
  pages={100940},
  year={2024},
  publisher={Elsevier}
}

@article{wijayathunga2023challenges,
  title={Challenges and solutions for autonomous ground robot scene understanding and navigation in unstructured outdoor environments: A review},
  author={Wijayathunga, Liyana and Rassau, Alexander and Chai, Douglas},
  journal={Applied Sciences},
  volume={13},
  number={17},
  pages={9877},
  year={2023},
  publisher={MDPI}
}

@article{simmons1994structured,
  title={Structured control for autonomous robots},
  author={Simmons, Reid G},
  journal={IEEE transactions on robotics and automation},
  volume={10},
  number={1},
  pages={34--43},
  year={1994},
  publisher={IEEE}
}

@INPROCEEDINGS{ICCBF,
  author={Agrawal, Devansh R. and Panagou, Dimitra},
  booktitle={2021 60th IEEE Conference on Decision and Control (CDC)}, 
  title={Safe Control Synthesis via Input Constrained Control Barrier Functions}, 
  year={2021},
  volume={},
  number={},
  pages={6113-6118},
  doi={10.1109/CDC45484.2021.9682938}}

@article{yang2023empc,
  title={{EMPC} with adaptive {APF} of obstacle avoidance and trajectory tracking for autonomous electric vehicles},
  author={Yang, Hongjiu and Wang, Zhengyu and Xia, Yuanqing and Zuo, Zhiqiang},
  journal={ISA transactions},
  volume={135},
  pages={438--448},
  year={2023},
  publisher={Elsevier}
}

@article{dynamics_differential,
  title={Adaptive rigidity-based formation control for multirobotic vehicles with dynamics},
  author={Cai, Xiaoyu and De Queiroz, Marcio},
  journal={IEEE Transactions on Control Systems Technology},
  volume={23},
  number={1},
  pages={389--396},
  year={2014},
  publisher={IEEE}
}

@article{barrier,
  title={3D UAV navigation with moving-obstacle avoidance using barrier Lyapunov functions},
  author={Restrepo, Esteban and Sarras, Ioannis and Loria, Antonio and Marzat, Julien},
  journal={IFAC-PapersOnLine},
  volume={52},
  number={12},
  pages={49--54},
  year={2019},
  publisher={Elsevier}
}

\end{document}